\documentclass[a4paper, fleqn]{article}

\usepackage[utf8]{inputenc}
\usepackage[T1]{fontenc}
\usepackage{epigraph}
\usepackage{lettrine}
\usepackage{type1ec}
\usepackage{latexsym}
\usepackage{amsmath}
\usepackage[makeroom]{cancel}
\usepackage{amsfonts}
\usepackage{amssymb}
\usepackage{amsthm}
\usepackage{mathrsfs}
\usepackage{setspace}
\usepackage{afterpage}
\usepackage{newlfont}
\usepackage{array}
\usepackage{tabularx}
\usepackage{dcolumn}
\usepackage{graphicx}
\usepackage{booktabs}
\usepackage[bf]{caption}
\usepackage{tikz}
\usepackage{xcolor}
\usetikzlibrary{arrows}
\usepackage{bussproofs}
\usepackage{stmaryrd}
\usepackage{multirow}

\graphicspath{ {images/} }

\author{Tiziano Dalmonte, Charles Grellois, Nicola Olivetti}
\title{\Intlogic s: \\ A general framework\footnote{Preliminary version. This work was partially supported by the Project TICAMORE ANR-16-CE91-0002-01.}}
\date{\begin{small}\emph{Aix Marseille Univ, Universit\'e de Toulon, CNRS, LIS, Marseille, France}\end{small}}

\theoremstyle{definition}
\newtheorem{theorem}{Theorem}[section]
\theoremstyle{definition}
\newtheorem{lemma}[theorem]{Lemma}
\theoremstyle{definition}
\newtheorem{fact}[theorem]{Fact}
\theoremstyle{definition}

\theoremstyle{definition}
\newtheorem{proposition}[theorem]{Proposition}
\theoremstyle{definition}
\newtheorem{corollary}[theorem]{Corollary}
\theoremstyle{definition}
\newtheorem{definition}{Definition}[section]
\theoremstyle{definition}
\newtheorem{example}{Example}[section]

\newcommand{\mc}{\mathcal}
\newcommand{\seq}{\Rightarrow}
\newcommand{\wk}{$\mathsf{wk}$}
\newcommand{\lwk}{$\mathsf{L}$\wk}
\newcommand{\rwk}{$\mathsf{R}$\wk}
\newcommand{\cut}{$\mathsf{Cut}$}
\newcommand{\ctr}{$\mathsf{ctr}$}
\newcommand{\efq}{\textsf{efq}}
\newcommand{\modusponens}{\textsf{mp}}

\newcommand{\hnbox}{\thickness{\textsf{{N$_\Box$}}}} 
\newcommand{\hndiam}{\thickness{\textsf{{N$_\diam$}}}} 
\newcommand{\hc}{\thickness{\textsf{{C}}}} 
\newcommand{\genznbox}{\hnbox} 
\newcommand{\genzndiam}{\hndiam}

\newcommand{\Estar}{\thickness{\textsf{E}}$^*$}
\newcommand{\E}{\thickness{\textsf{E}}}
\newcommand{\EN}{\thickness{\textsf{EN}}}
\newcommand{\EC}{\thickness{\textsf{EC}}}
\newcommand{\ECN}{\thickness{\textsf{ECN}}}
\newcommand{\EM}{\thickness{\textsf{M}}}
\newcommand{\EMN}{\thickness{\textsf{MN}}}
\newcommand{\EMC}{\thickness{\textsf{MC}}}
\newcommand{\EMCN}{\thickness{\textsf{MCN}}}
\newcommand{\K}{\thickness{\textsf{K}}}

\newcommand{\smallone}{$_\textup{1}$}
\newcommand{\smalltwo}{$_\textup{2}$}
\newcommand{\smallthree}{$_\textup{3}$}

\newcommand{\thickness}{}

\newcommand{\intuitionistic}{\thickness{\textsf I}}
\newcommand{\unoE}{\intuitionistic\thickness{\textsf{E\smallone}}} 
\newcommand{\dueE}{\intuitionistic\thickness{\textsf{E\smalltwo}}}   
\newcommand{\treE}{\intuitionistic\thickness{\textsf{E\smallthree}}} 
\newcommand{\unoM}{\intuitionistic\thickness{\textsf{M}}}        
\newcommand{\unoENbox}{\intuitionistic\thickness{\textsf{E\smallone N$_\Box$}}} 
\newcommand{\unoENdiam}{\intuitionistic\thickness{\textsf{E\smallone N$_\diam$}}} 
\newcommand{\dueENbox}{\intuitionistic\thickness{\textsf{E\smalltwo N$_\Box$}}}
\newcommand{\dueENdiam}{\intuitionistic\thickness{\textsf{E\smalltwo N$_\diam$}}}
\newcommand{\treENbox}{\intuitionistic\thickness{\textsf{E\smallthree N$_\Box$}}}
\newcommand{\treENdiam}{\intuitionistic\thickness{\textsf{E\smallthree N$_\diam$}}}
\newcommand{\MNbox}{\intuitionistic\thickness{\textsf{MN$_\Box$}}} 
\newcommand{\MNdiam}{\intuitionistic\thickness{\textsf{MN$\diamond$}}}

\newcommand{\unoEC}{\intuitionistic\thickness{\textsf{E\smallone C}}} 
\newcommand{\dueEC}{\intuitionistic\thickness{\textsf{E\smalltwo C}}}   
\newcommand{\treEC}{\intuitionistic\thickness{\textsf{E\smallthree C}}} 
\newcommand{\unoMC}{\intuitionistic\thickness{\textsf{MC}}}        
\newcommand{\unoENboxC}{\intuitionistic\thickness{\textsf{E\smallone CN$_\Box$}}} 
\newcommand{\unoENdiamC}{\intuitionistic\thickness{\textsf{E\smallone CN$_\diam$}}} 
\newcommand{\dueENboxC}{\intuitionistic\thickness{\textsf{E\smalltwo CN$_\Box$}}}
\newcommand{\dueENdiamC}{\intuitionistic\thickness{\textsf{E\smalltwo CN$_\diam$}}}
\newcommand{\treENboxC}{\intuitionistic\thickness{\textsf{E\smallthree CN$_\Box$}}}
\newcommand{\treENdiamC}{\intuitionistic\thickness{\textsf{E\smallthree CN$_\diam$}}}
\newcommand{\MNboxC}{\intuitionistic\thickness{\textsf{MCN$_\Box$}}} 
\newcommand{\MNdiamC}{\intuitionistic\thickness{\textsf{MCN$_\diam$}}}

\newcommand{\iboxE}{\ibox\E}
\newcommand{\iboxEN}{\ibox\EN}
\newcommand{\iboxEC}{\ibox\EC}
\newcommand{\iboxECN}{\ibox\ECN}
\newcommand{\iboxEM}{\ibox\EM}
\newcommand{\iboxEMN}{\ibox\EMN}
\newcommand{\iboxEMC}{\ibox\EMC}
\newcommand{\iboxEMCN}{\ibox\EMCN}
\newcommand{\idiamE}{\idiam\E}
\newcommand{\idiamEN}{\idiam\EN}

\newcommand{\idiamEM}{\idiam\EM}
\newcommand{\idiamEMN}{\idiam\EMN}

\newcommand{\gentzen}{\textsf{G}.}

\newcommand{\g}{$\mathsf{G^*}$}
\newcommand{\gunoE}{\gentzen\unoE}
\newcommand{\gdueE}{\gentzen\dueE}    
\newcommand{\gtreE}{\gentzen\treE}  
\newcommand{\gunoM}{\gentzen\unoM}

\newcommand{\logichedue}{\dueE(\axC,\axNdiam,\axNbox)}
\newcommand{\calcolidue}{\gentzen\logichedue}

\newcommand{\gC}{$\mathsf{GC}^*$}
\newcommand{\gunoEC}{\gentzen\unoEC} 
\newcommand{\gdueEC}{\gentzen\dueEC}    
\newcommand{\gtreEC}{\gentzen\treEC}  
\newcommand{\gunoMC}{\gentzen\unoMC}

\newcommand{\gMNboxC}{\gentzen\MNboxC}

\newcommand{\giboxE}{\gentzen\iboxE}
\newcommand{\giboxEN}{\gentzen\iboxEN}
\newcommand{\giboxEC}{\gentzen\iboxEC}
\newcommand{\giboxECN}{\gentzen\iboxECN}
\newcommand{\giboxEM}{\gentzen\iboxEM}
\newcommand{\giboxEMN}{\gentzen\iboxEMN}
\newcommand{\giboxEMC}{\gentzen\iboxEMC}
\newcommand{\giboxEMCN}{\gentzen\iboxEMCN}
\newcommand{\gidiamE}{\gentzen\idiamE}
\newcommand{\gidiamEN}{\gentzen\idiamEN}

\newcommand{\gidiamEM}{\gentzen\idiamEM}
\newcommand{\gidiamEMN}{\gentzen\idiamEMN}

\newcommand{\h}{$\mathsf{H^*}$}

\newcommand{\axMdiam}{$\mathsf M_\Diamond$}
\newcommand{\axCdiam}{$\mathsf C_\Diamond$}
\newcommand{\axNdiam}{$\mathsf N_\Diamond$}
\newcommand{\axMbox}{$\mathsf M_\Box$}
\newcommand{\axCbox}{$\mathsf C_\Box$}
\newcommand{\axNbox}{$\mathsf N_\Box$}
\newcommand{\dualbox}{$\mathsf{Dual_\Box}$}
\newcommand{\dualdiam}{$\mathsf{Dual_\Diamond}$}
\newcommand{\axKdiam}{$\mathsf K_\Diamond$}
\newcommand{\axKbox}{$\mathsf K_\Box$}

\newcommand{\axM}{\thickness{\textsf M}}
\newcommand{\axC}{\thickness{\textsf C}}
\newcommand{\axN}{\thickness{\textsf N}}
\newcommand{\axT}{\thickness{\textsf T}}
\newcommand{\axD}{\thickness{\textsf D}}
\newcommand{\axquattro}{\thickness{\textsf 4}}
\newcommand{\axcinque}{\thickness{\textsf 5}}

\newcommand{\rebox}{$\mathsf{E_\Box}$}
\newcommand{\rediam}{$\mathsf{E_\Diamond}$}
\newcommand{\rmbox}{$\mathsf{Mon_\Box}$}
\newcommand{\rmdiam}{$\mathsf{Mon_\Diamond}$}
\newcommand{\rulenbox}{$\mathsf{Nec}$}

\newcommand{\intunoa}{$\mathsf{weak_a}$}
\newcommand{\intunob}{$\mathsf{weak_b}$}
\newcommand{\intduea}{$\mathsf{neg_a}$}
\newcommand{\intdueb}{$\mathsf{neg_b}$}
\newcommand{\inttre}{$\mathsf{str}$}

\newcommand{\apiceseq}{seq}
\newcommand{\grebox}{$\mathsf{E_\Box^{\apiceseq}}$}
\newcommand{\grediam}{$\mathsf{E_\Diamond^{\apiceseq}}$}
\newcommand{\grmbox}{$\mathsf{M_\Box^{\apiceseq}}$}
\newcommand{\grmdiam}{$\mathsf{M_\Diamond^{\apiceseq}}$}
\newcommand{\grulenbox}{$\mathsf{N_\Box^{\apiceseq}}$}
\newcommand{\grulendiam}{$\mathsf{N_\Diamond^{\apiceseq}}$}
\newcommand{\gintunoa}{$\mathsf{weak_a^{\apiceseq}}$}
\newcommand{\gintunob}{$\mathsf{weak_b^{\apiceseq}}$}
\newcommand{\gintduea}{$\mathsf{neg_a^{\apiceseq}}$}
\newcommand{\gintdueb}{$\mathsf{neg_b^{\apiceseq}}$}
\newcommand{\ginttre}{$\mathsf{str^{\apiceseq}}$}
\newcommand{\greboxc}{$\mathsf{E_\Box C^{\apiceseq}}$}

\newcommand{\grmboxc}{$\mathsf{M_\Box C^{\apiceseq}}$}

\newcommand{\gintunobc}{$\mathsf{weak_bC^{\apiceseq}}$}
\newcommand{\gintdueac}{$\mathsf{neg_aC^{\apiceseq}}$}
\newcommand{\gintduebc}{$\mathsf{neg_bC^{\apiceseq}}$}
\newcommand{\ginttrec}{$\mathsf{strC^{\apiceseq}}$}

\newcommand{\inseq}{$\mathsf{Ax}$}
\newcommand{\lbot}{$\mathsf{L\bot}$}
\newcommand{\lland}{$\mathsf{L\land}$}
\newcommand{\rland}{$\mathsf{R\land}$}
\newcommand{\llor}{$\mathsf{L\lor}$}
\newcommand{\rlor}{$\mathsf{R\lor}$}
\newcommand{\limp}{$\mathsf{L}$$\mathsf{\imp}$}
\newcommand{\rimp}{$\mathsf{R}$$\mathsf{\imp}$}

\newcommand{\rneg}{$\mathsf{R\neg}$}

\newcommand{\X}{\logic}
\newcommand{\gX}{\gentzen\X}

\newcommand{\gtrei}{$\mathsf{G3ip}$}
\newcommand{\gtre}{$\mathsf{G3}$}
\newcommand{\il}{\textsf{IPL}}
\newcommand{\cl}{\textsf{CPL}}

\newcommand{\diam}{\Diamond}
\newcommand{\ax}{\AxiomC}
\newcommand{\uinf}{\UnaryInfC}
\newcommand{\binf}{\BinaryInfC}
\newcommand{\llab}{\LeftLabel}
\newcommand{\rlab}{\RightLabel}
\newcommand{\disp}{\DisplayProof}

\newcommand{\W}{\mathcal W}
\newcommand{\nbox}{\mathcal N_\Box}
\newcommand{\ndiam}{\mathcal N_\Diamond}
\newcommand{\V}{\mathcal V}
\newcommand{\M}{\mathcal M}
\newcommand{\N}{\mathcal N}
\newcommand{\less}{\preceq}
\newcommand{\more}{\succeq}
\newcommand{\R}{\mathcal R}
\newcommand{\Nk}{\mathcal N_k}
\newcommand{\Mk}{\mathcal M_k}
\newcommand{\rel}{\mathscr R}

\newcommand{\Mc}{\mathcal M^c}
\newcommand{\Vc}{\mathcal V^c}
\newcommand{\lessc}{\preceq^c}
\newcommand{\Nc}{\mathcal N^c}
\newcommand{\Wc}{\mathcal W^c}
\newcommand{\Mcplus}{\mathcal M^c_+}
\newcommand{\McHW}{\mathcal M^c_{\textup{\HW}}}
\newcommand{\McCK}{\mathcal M^c_{\textup{\CK}}}
\newcommand{\wmodel}{\HW-model}

\newcommand{\Wstar}{\W^*}
\newcommand{\Mstar}{\M^*}
\newcommand{\Vstar}{\V^*}
\newcommand{\Rstar}{\R^*}
\newcommand{\lessstar}{\less^*}
\newcommand{\morestar}{\more^*}
\newcommand{\f}{\mathbf f}
\newcommand{\ff}{(\f, \{\f\})}

\newcommand{\longto}{\longrightarrow}

\newcommand{\pow}{\mathcal P}

\newcommand{\atm}{Atm}
\newcommand{\sbf}{Sbf}

\newcommand{\G}{\Gamma}

\newcommand{\up}{\uparrow_{pr}$$}

\newcommand{\logic}{\thickness{\textsf L}}
\newcommand{\logicone}{\logic$_1$}
\newcommand{\logictwo}{\logic$_2$}
\newcommand{\glogic}{\gentzen\logic}

\newcommand{\vd}{\vdash}
\newcommand{\Vd}{\Vdash}

\newcommand{\Vdr}{\Vdash_{r}}
\newcommand{\Vdk}{\Vdash_{k}}
\newcommand{\Vdn}{\Vd}

\newcommand{\setneg}{-}
\newcommand{\pos}{^+}

\newcommand{\wij}{Wijesekera}

\newcommand{\imp}{\supset}
\newcommand{\coimp}{\subset}
\newcommand{\monomodalbox}{\ibox\E$^*$}
\newcommand{\monomodaldiam}{\idiam\E$^*$}
\newcommand{\ibox}{$\Box$-\intuitionistic}
\newcommand{\idiam}{$\diam$-\intuitionistic}

\newcommand{\lan}{\mathcal L}
\newcommand{\ldiam}{\mc L_\diam}
\newcommand{\lbox}{\mc L_\Box}

\newcommand{\connuno}{$weakInt$}
\newcommand{\Connuno}{$weakInt$}
\newcommand{\connduea}{$negInt_a$}
\newcommand{\conndueb}{$negInt_b$}

\newcommand{\Connduea}{$negInt_a$}
\newcommand{\Conndueb}{$negInt_b$}
\newcommand{\Conndue}{$negInt$}
\newcommand{\conntre}{$strInt$}
\newcommand{\Conntre}{$strInt$}

\newcommand{\Connunobis}{Weak interaction}

\newcommand{\Conntrebis}{Strong interaction}
\newcommand{\Connduebisaaux}{Negation closure int\_a}
\newcommand{\Connduebisbaux}{Negation closure int\_b}

\newcommand{\wcondition}{$WInt$}
\newcommand{\wconditionbis}{$WInt'$}

\newcommand{\lessclosure}{$\less$-closure}

\newcommand{\nboxc}{\nbox^c}
\newcommand{\ndiamc}{\ndiam^c}
\newcommand{\nboxplus}{\nbox^+}
\newcommand{\ndiamplus}{\ndiam^+}

\newcommand{\GW}{\gentzen\HW}
\newcommand{\GCK}{\gentzen\CK}

\newcommand{\Wrule}{$\mathsf{W^{seq}}$}

\newcommand{\HW}{\CCDL$^{\mathsf p}$}
\newcommand{\CK}{\thickness{\textsf{CK}}}
\newcommand{\CCDL}{\thickness{\textsf{CCDL}}}
\newcommand{\PLL}{\thickness{\textsf{PLL}}}

\newcommand{\wclass}{w_\sim}
\newcommand{\vclass}{v_\sim}

\newcommand{\alphaclass}{\alpha_\sim}
\newcommand{\Aclass}{[A]^\sim}
\newcommand{\Aunoclass}{[A_1]^\sim}
\newcommand{\Anclass}{[A_n]^\sim}
\newcommand{\Bclass}{[B]^\sim}

\newcommand{\Mf}{\M^*}
\newcommand{\Wf}{\W^*}

\newcommand{\lessf}{\less^*}
\newcommand{\Vf}{\V^*}
\newcommand{\nboxf}{\nbox^*}
\newcommand{\ndiamf}{\ndiam^*}

\newcommand{\Mcirc}{\M^\circ}
\newcommand{\nboxcirc}{\nbox^\circ}
\newcommand{\ndiamcirc}{\ndiam^\circ}

\newcommand{\Mn}{\M_n}
\newcommand{\Mr}{\M_r}

\newcommand{\w}{w}

\newcommand{\reqone}{R1}
\newcommand{\reqtwo}{R2}
\newcommand{\reqthree}{R3}

\newcommand{\cupledintmodel}{coupled intuitionistic neighbourhood model}
\newcommand{\Intmodel}{Intuitionistic neighbourhood model}
\newcommand{\komodel}{Kojima's model}
\newcommand{\intlogic}{intuitionistic non-normal modal logic}
\newcommand{\Intlogic}{Intuitionistic non-normal modal logic}
\newcommand{\intmonologic}{intuitionistic non-normal monomodal logic}
\newcommand{\intbilogic}{intuitionistic non-normal bimodal logic}
\newcommand{\Intbilogic}{Intuitionistic non-normal bimodal logic}
\newcommand{\boxmodel}{$\Box$-INM}
\newcommand{\diammodel}{$\diam$-INM}
\newcommand{\intmodel}{CINM}

\begin{document}
\maketitle

\begin{abstract}
We define a family of intuitionistic non-normal modal logics;
they can bee seen as intuitionistic counterparts of classical ones.
We first consider monomodal logics, which contain only one between Necessity and Possibility.
We then consider the more important case of bimodal logics,
which contain both modal operators.
In this case we define several interactions between Necessity and Possibility of 
increasing strength, although weaker than duality.
For all logics we provide both a Hilbert axiomatisation and a cut-free sequent calculus,
on its basis we also prove their decidability.
We then give a semantic characterisation of our logics in terms of  
neighbourhood models. 
Our semantic framework captures modularly not only our systems 
but also already known intuitionistic non-normal modal logics such as Constructive K (CK) and
the propositional fragment of Wijesekera's Constructive Concurrent Dynamic Logic.
\end{abstract}

\section{Introduction}

Both intuitionistic modal logic and non-normal modal logic have been studied for a long time. 
The  study of modalities with an intuitionistic basis goes back to Fitch in the late 40s (Fitch \cite{Fitch})
 and has led to an important stream of research. We can very schematically identify two  traditions: so-called Intuitionistic modal logics \emph{versus} Constructive modal logics. Intuitionistic modal logics have been  systematised  by Simpson \cite{Simpson}, whose main goal is to define an analogous of classical modalities justified from an intuitionistic point of view. On the other hand, constructive modal logics are  mainly motivated by their applications to computer science, such as the type-theoretic interpretations (Curry--Howard correspondence, typed lambda calculi), verification and knowledge representation,\footnote{For a recent survey see Stewart \emph{et al.} \cite{Stewart} and references therein.}
 but also by their mathematical semantics (Goldblatt \cite{Goldblatt}).

On the other hand, non-normal modal logics have been strongly motivated on a philosophical and epistemic ground. They are called ``non-normal'' as they do not satisfy
all the axioms and rules of the minimal normal modal logic \K. They have been studied since the seminal works of Scott, Lemmon, and Chellas (\cite{Scott}, \cite{Chellas}, see Pacuit \cite{Pacuit}  for a survey), 
 and can be seen as generalisations of standard modal logics. They  have found an interest in several areas such as epistemic and deontic reasoning, reasoning about games, and reasoning about probabilistic notions such as ``truth in most of the cases''.

Although the two areas have grown up  seemingly without any interaction, it can be noticed that some intuitionistic or constructive  modal logics investigated in the literature contain non-normal modalities. The prominent example is  the logic \CCDL{} proposed by \wij{} \cite{\wij}, whose  propositional fragment
(that we call \HW) has been recently investigated by Kojima \cite{Kojima}. This logic has a normal $\Box$ modality and a non-normal $\Diamond$ modality, where  $\Diamond$ does not distribute over the $\lor$, that is 
$$(\textup{\axCdiam}) \quad\quad \Diamond(A \lor B)\imp \Diamond A \lor \Diamond B$$
is not valid. The original motivation by \wij{} comes from Constructive Concurrent Dynamic Logic, but the logic has also an interesting epistemic interpretation in terms of internal/external observers proposed by Kojima. A related system is Constructive \K{} (\CK), that has been  proposed by Bellin \emph{et al.} \cite{Bellin} and further investigated by Mendler and de Paiva \cite{Mendler1}, Mendler and Scheele \cite{Mendler2}. 
This system not only rejects \axCdiam, but also its nullary version $\Diamond \bot \imp \bot$ (\axNdiam). 
In contrast all these systems assume a normal interpretation of $\Box$ so that 
$$\Box (A \land B) \imp\coimp (\Box A \land \Box B)$$
is always assumed. 
A further example is Propositional Lax Logic (\PLL) by  Fairtlough and Mendler \cite{Fairtlough}, an intuitionistic monomodal logic for hardware verification where the modality 
does not validate the rule of necessitation.

Finally, all intuitionistic modal logics reject the interdefinability of the two operators:
 $$\Box A \imp\coimp\neg \Diamond \neg A$$
 and its boolean equivalents.

To the best of our knowledge, no systematic investigation of non-normal modalities with an intuitionistic base has been carried out so far.  Our aim is to lay down a general framework which can accommodate in a uniform way intuitionistic counterparts of the classical cube of non-normal modal logics, as well as \HW{} and  \CK{} mentioned above. As we shall see, the adoption of an intuitionistic base leads to a finer analysis of non-normal modalities than in the classical case.  In addition to the motivations for classical non-normal modal logics briefly recalled above,  an intutionistic interpretation of non-normal modalities may be justified by more  specific interpretations, of which we mention two:
\begin{itemize}
	\item The \textbf{deontic interpretation}: The standard interpretation of deontic operators $\Box$ (Obligatory), 
$\Diamond$ (Permitted) is normal: but it has been known for a long time that the normal interpretation is problematic when dealing for instance with ``Contrary to duty obligations".%
\footnote{For a survey on puzzles related to a normal interpretation of the deontic modalities see McNamara \cite{McNamara}.}
One solution is to adopt a non-normal interpretation, rejecting in particular the monotonicity principle (from $A \imp B$ is valid infer $\Box A\imp\Box B$). 
Moreover, a constructive reading of the deontic modalities would further reject their interdefinability: 
one may require that the permission of $A$ must be justified explicitly or positively (say by a proof from a corpus of norms) and not just established by the fact that $\neg A$ is not obligatory
(see for instance the distinction between weak and strong permissions in von Wright \cite{Wright}).
	\item The \textbf{contextual interpretation}: A contextual reading of the modal operators  is proposed in Mendler and de Paiva \cite{Mendler1}. 
In this interpretation $\Box A$ is read as ``$A$ holds in all contexts'' and $\Diamond A$ as ``A holds in some context''. This interpretation invalidates \axCdiam, while retaining the distribution of $\Box$ over conjunction (\axCbox). But this contextual interpretation is not the only possible one. 
	We can interpret $\Box A$ as $A$ is ``justified'' (proved) in some context $c$, no matter what is meant by a context (for instance a knowledge base), and $\Diamond A$ as $A$ is ``compatible'' (consistent) with every context. With this interpretation both operators would be non-normal as they would satisfy neither \axCbox, nor \axCdiam. 
\end{itemize}
As we said, our aim is to provide a general framework for non-normal modal logics with an intuitionistic base. However, in order to identify and restrain  the family of logics of interest, we adopt some criteria, which partially coincide with Simpson's requirements (Simpson \cite{Simpson}):
 \begin{itemize}
 	\item The modal logics should be conservative extensions of \il.
 	\item The disjunction property must hold.
 	\item The two modalities should not be interdefinable.
 	\item We do not consider systems containing the controversial \axCdiam. 
 \end{itemize}
Our starting point is the study of \emph{monomodal} systems, which extend \il{} with either $\Box$ or $\Diamond$, but not both. We consider the monomodal logics  corresponding to the classical cube  generated by the weakest logic \E{} extended with conditions \textsf{M, N, C} 
(with the exception of \axCdiam). 
We give an axiomatic characterisation of these logics  and equivalent cut-free sequent systems similar to the one by Lavendhomme and Lucas \cite{Lavendhomme} for the classical case.

Our main interest is however in logics which contain \emph{both} $\Box$ and $\Diamond$, and allow some form of interaction between the two. Their interaction is always  weaker than  interdefinability. 
In order to define logical systems we take a proof-theoretical perspective:  the existence of a simple cut-free system, as in the monomodal case,  is our criteria to identify meaningful systems. \emph{A system is retained if the combination of sequent rules amounts to a cut-free system}. 

It turns out that one can distinguish \emph{three} degrees of interaction between  $\Box$ and $\Diamond$, that are determined by answering to the question,  for any two formulas $A$ and $B$:
\begin{quote}
\emph{under what conditions $\Box A$ and $\Diamond B$ are jointly inconsistent?}
\end{quote}
Since there are \emph{three} degrees of interaction,  even the weakest classical logic \E{} has \emph{three} intuitionistic counterparts of increasing strength. When combined with \textsf{M, N, C} properties of the classical cube, we end up with a family of 24 distinct systems, all enjoying a cut-free calculus  and, as we prove, an equivalent Hilbert axiomatisation. This shows that intuitionistic non-normal modal logic allows for finer distinctions  whence a richer theory  than in the classical case.

The existence of a cut-free calculus for each of the logics  has some important consequences: We can prove that all systems are indeed distinct, that all of them are ``good'' extensions of intuitionistic logic, and more importantly that all of them are decidable.

We then tackle the problem of giving a semantic characterisation of this family of logics. The natural setting is to consider an intuitionistic version of neighbourhood models for classical logics. Since we want to deal with the language containing both $\Box$ and $\diam$, we consider neighbourhood models containing \emph{two} distinct neighbourhood functions $\nbox$ and $\ndiam$. 
As in standard intuitionistic models, they also contain a partial order on worlds. Different forms of interaction between the two modal operators correspond to different (but natural) conditions relating the two neighbourhood functions. 
By considering further closure conditions of neighbourhoods, analogous to the classical case,  we can show that this semantic characterises \emph{modularly} the full family of logics.  Moreover we prove, through a filtration argument, 
that most of the logics have the \emph{finite model property}, thereby obtaining a semantic proof of their decidability.

It is worth noticing that in the (easier) case of 
intuitionistic monomodal logic with only $\Box$ 
a similar semantics and a matching completeness theorem have been given by Goldblatt \cite{Goldblatt}.
More recently, Goldblatt's semantics for the intuitionistic version of 
system \E{} has been reformulated and extended to axiom \axT{} by Witczak \cite{Witczak2}.

But  our neighbourhood models have a wider application than the characterisation of the family of logics mentioned above. We show that adding suitable  \emph{interaction conditions} between $\nbox$ and $\ndiam$ we can capture \HW{} as well as \CK. We show this fact first directly by proving that both \HW{}  and \CK{}  are sound and complete with respect to our models satisfying  an additional condition.
We then prove the same result by relying on some pre-existing semantics of these two logics and by transforming models. In case of \HW, there exists already a characterisation of it in terms of neighbourhood models, given by Kojima \cite{Kojima}, although the type of models is different, in particular Kojima's models contain only one neighbourhood function. 

The case of \CK{} is more complicated, whence more interesting: this logic is  characterised  by a relational semantics defined in terms of  Kripke models of a \emph{peculiar} nature: they  contain ``fallible'' worlds, \emph{i.e.} worlds  which force $\bot$. We are able to show directly that relational models can be transformed into our neighbourhood models satisfying a specific interaction condition and \emph{vice versa}. 

All in all,  we get that the well-known \CK{} can be characterised by  neighbourhood models, after all rather standard structures, alternative to non-standard Kripke models with fallible worlds. This fact  provides further evidence in favour of our neighbourhood semantics as a versatile tool to analyse intuitionistic non-normal modal logics.

\section{Classical non-normal modal logics}\label{section classical}
\subsection{Hilbert systems}
Classical non-normal modal logics are defined on a propositional modal language $\lan$
based on a set $\atm$ of countably many propositional variables.
Formulas are given by the following grammar, where $p$ ranges over $\atm$:
\begin{center}
$A ::= p \mid \bot \mid A\land A \mid A\lor A \mid A\imp A \mid \Box A \mid \diam A$.
\end{center}

We use $A, B, C$ as metavariables for formulas of $\lan$.
$\top$, $\neg A$ and $A\imp\coimp B$ are abbreviations for, respectively, $\bot\imp\bot$, $A\imp\bot$ and $(A\imp B)\land(B\imp A)$.
We take both modal operators $\Box$ and $\diam$ as primitive
(as well as all boolean connectives), 
as it will be convenient for the intuitionistic case. 
Their duality in classical modal logics is recovered by adding to any system one of the duality axioms
\dualbox{} or \dualdiam{} (Figure \ref{modal axioms}),
which are equivalent in the classical setting.

\begin{figure}
\noindent
\fbox{\begin{minipage}{34em}

\vspace{0.1cm}

\begin{tabular}{l l l l l}

\multicolumn{5}{l}{ \vspace{0.2cm}\textbf{\emph{a.} Modal axioms and rules defining non-normal modal logics}} \\

\vspace{0.2cm}
\rebox 
&
\ax{$A\imp B$}\ax{$B \imp A$}
\binf{$\Box A\imp \Box B$}\disp 
&
&
\rediam 
&
\ax{$A\imp B$}\ax{$B\imp A$}
\binf{$\diam A\imp \diam B$}\disp \\

\vspace{0.2cm}
\axMbox & $\Box (A\land B) \imp \Box A$ && \axMdiam & $\diam A \imp \diam(A\lor B)$ \\

\vspace{0.2cm}
\axCbox & $\Box A \land \Box B \imp \Box (A\land B)$  && \axCdiam & $\diam(A\lor B) \imp \diam A \lor \diam B$\\

\vspace{0.5cm}
\axNbox  & $\Box\top$ && \axNdiam & $\neg\diam\bot$ \\

\multicolumn{5}{l}{\vspace{0.2cm} \textbf{\emph{b.} Duality axioms}} \\

\vspace{0.5cm}
\dualbox & $\diam A \imp\coimp \neg\Box\neg A$ && \dualdiam & $\Box A \imp\coimp \neg\diam\neg A$ \\

\multicolumn{5}{l}{\vspace{0.2cm}\textbf{\emph{c.} Further relevant modal axioms and rules}} \\

\vspace{0.2cm}
\axKbox & $\Box(A \imp B) \imp (\Box A \imp \Box B)$ &\qquad&
\axKdiam & $\Box(A \imp B) \imp (\diam A \imp \diam B)$ \\

\multicolumn{2}{l}{\rulenbox \ \ \ax{$A$} \uinf{$\Box A$} \disp
\qquad \ \rmbox \ \ax{$A\imp B$} \uinf{$\Box A \imp \Box B$} \disp} 
&& 
\rmdiam & \ax{$A\imp B$} \uinf{$\diam A \imp \diam B$} \disp \\
\end{tabular}

\end{minipage}}
\caption{\label{modal axioms} Modal axioms.}  
\end{figure}

The weakest classical non-normal modal logic \E{} is defined in language $\lan$
by extending classical propositional logic (\cl)  with a duality axiom and rule \rebox{},
and it can be extended further 
by adding any combination of axioms \axMbox, \axCbox{} and \axNbox.
We obtain in this way eight distinct systems (Figure \ref{classical cube}),
which compose the family of classical non-normal modal logics.

Equivalent axiomatisations for these systems 
are given by considering the modal axioms in the right-hand column of Figure \ref{modal axioms}($a$).
Thus, logic \E{} could be defined by extending \cl{} with axiom \dualbox{} and rule \rediam,
and its extensions are given by adding combinations of axioms \axMdiam, \axCdiam{} and \axNdiam.

It is worth recalling that axioms \axMbox, \axMdiam{} and \axNbox{}
are syntactically equivalent with the rules \rmbox, \rmdiam{} and \rulenbox, respectively,
and that axiom \axKbox{} is derivable from \axMbox{} and \axCbox.
As a consequence, we have that the top system \EMCN{} 
is equivalent to the weakest classical normal modal logic \K.

\begin{figure}
\qquad \hfill
\resizebox{4cm}{3cm}{
\begin{tikzpicture}
	\node (a) at  (0,0)  {\E};
    \node (b) at (0, 2.1) {\EM};
    \node  (c) at (-1.5, 0.7) {\EC};
    \node (d) at (2.3, 0.7) {\EN};
    \node (e) at (-1.5, 2.8) {\EMC};
    \node (f) at (2.3, 2.8) {\EMN};
    \node (g) at (0.8, 1.4) {\ECN};
    \node (h) at (0.8, 3.5) {\EMCN{} (\K)};

\draw (a) -- (b);
\draw (a) -- (c);
\draw (a) -- (d);
\draw (b) -- (e);
\draw (b) -- (f);
\draw (c) -- (e);
\draw [dashed] (c) -- (g);
\draw (d) -- (f);
\draw [dashed] (d) -- (g);
\draw (e) -- (h);
\draw (f) -- (h);
\draw [dashed] (g) -- (h);
\end{tikzpicture}
}
\hfill \qquad
\caption{\label{classical cube} The classical cube.}
\end{figure}
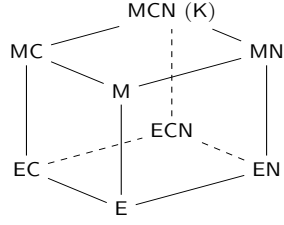

\subsection{Neighbourhood semantics}
The standard semantics for classical non-normal modal logics is based on the so-called 
neighbourhood (or minimal, or Scott-Montague) models.

\begin{definition}\label{classical neighbourhood models}
A \emph{neighbourhood model} is a triple $\M=\langle\W,\N,\V\rangle$,
where  $\W$ is a non-empty set,
$\N$ is a neighbourhood function $\W\longto\pow(\pow(\W))$,
and $\V$ is a  valuation function $\W\longto\atm$.
A neighbourhood model is supplemented, closed under intersection, or contains the unit,
if $\N$ satisfies the following properties:

\vspace{0.2cm}
\begin{tabular}{l l}
If $\alpha\in\N(w)$ and $\alpha\subseteq\beta$, then $\beta\in\N(w)$ & (Supplementation);\\

If $\alpha,\beta\in\N(w)$, then $\alpha\cap\beta\in\N(w)$ & (Closure under intersection);\\

$\W\in\N(w)$ \ for all $w\in \W$ & (Containing the unit).\\
\end{tabular}

\vspace{0.2cm}
\noindent
The forcing relation $w\Vd A$ is defined inductively as follows:

\vspace{0.2cm}
\begin{tabular}{l l l}
$w\Vd p$ & iff & $p\in\V(w)$; \\
$w\not\Vd \bot$; \\
$w\Vd B\land C$ & iff & $w\Vd A$ and $w\Vd B$; \\
$w\Vd B\lor C$ & iff & $w\Vd A$ or $w\Vd B$; \\
$w\Vd B \imp C$ & iff & $w \Vd B$ implies $w \Vd C$; \\
$w\Vd\Box B$ & iff & $[B]\in\N(w)$; \\
$w\Vd\diam B$ & iff & $\W\setminus[B]\notin\N(w)$; \\
\end{tabular}

\vspace{0.2cm}
\noindent
where $[B]$ denotes the set $\{v\in\W \mid v\Vd B\}$, 
called the \emph{truth set} of $B$.
\end{definition}

We can also recall that in the supplemented case, the forcing conditions for modal formulas are equivalent to the following ones:

\vspace{0.2cm}
\begin{tabular}{l l l}
$w\Vd\Box B$ & iff & there is $\alpha\in\N(w)$ s.t.~$\alpha\subseteq[B]$; \\
$w\Vd\diam B$ & iff & for all $\alpha\in\N(w)$, $\alpha\cap[B]\not=\emptyset$. \\
\end{tabular}

\begin{theorem}[Chellas \cite{Chellas}]
Logic \E(\axM,\axC,\axN) is sound and complete with respect to neighbourhood models
(which in addition are supplemented, closed under intersection and contain the unit).
\end{theorem}

\section{Intuitionistic non-normal monomodal logics}\label{section monomodal}

Our definition of \intlogic s{} begins with
monomodal logics, that is logics containing only one modality, either $\Box$ or $\diam$.
We first define the axiomatic systems, and then present their sequent calculi.

Under ``intuitionistic modal logics'' we understand any modal logic \logic{} 
that
extends intuitionistic propositional logic (\il{})
and satisfies the following requirements:

\begin{itemize}
\item[(\reqone)]  
\logic{} is conservative over \il: its non-modal fragment coincides with \il.

\item[(\reqtwo)] 
\logic{} satisfies the disjunction property:
if $A\lor B$ is derivable, then at least one formula
between $A$ and $B$ is also derivable.
\end{itemize}

\subsection{Hilbert systems}

From the point of view of axiomatic systems, 
two different classes of \intmonologic s
can be defined by analogy with the definition of classical non-normal modal logics (cf. Section \ref{section classical}).
Intuitionistic modal logics are modal extensions of \il, for which we consider the following axiomatisation:

\vspace{0.2cm}
\noindent
\begin{tabular}{l l l l l}
$\imp$-$1$ & $A \imp (B \imp A)$ && $\land$-$1$ & $A\land B \imp A$ \\
$\imp$-$2$ & $(A \imp (B \imp C)) \imp ((A \imp B) \imp (A \imp C))$ && $\land$-$2$ & $A \land B \imp B$\\
$\lor$-$1$ & $A \imp A\lor B$ && $\land$-$3$ & $A \imp (B \imp A \land B)$ \\
$\lor$-$2$ & $B \imp A \lor B$ && \efq & $\bot \imp A$ \\
$\lor$-$3$ & $(A \imp C) \imp ((B \imp C) \imp (A \lor B \imp C))$ && \modusponens & \ax{$A$}\ax{$A\imp B$}\binf{$B$}\disp \\
\end{tabular}

 \vspace{0.2cm}
We define over \il{}  two families of \intmonologic s, 
that depend on the considered modal operator, and
are called therefore the $\Box$- and the $\diam$-family.
The $\Box$-family is defined in language $\lbox := \lan\setminus\{\diam\}$
by adding to \il{} the rule \rebox{} and any combination of axioms \axMbox, \axCbox{} and \axNbox.
The $\diam$-family is instead defined in language $\ldiam := \lan\setminus\{\Box\}$
by adding to \il{} the rule \rediam{} and any combination of axioms  \axMdiam{} and \axNdiam.
It is work remarking that 
we don't consider intuitionistic non-normal modal logics containing axiom \axCdiam.
We denote the resulting logics by, respectively,
\monomodalbox{} and \monomodaldiam, 
where \Estar{} replaces any system of the classical cube 
(for $\diam$-logics, any system non containing \axCdiam).

Notice that, having rejected the definability of the lacking modality,
$\Box$- and $\diam$-logics are distinct,
as $\Box$ and $\diam$ behave differently.
Moreover, as a consequence of the fact 
that the systems in the classical cube are pairwise non-equivalent, 
we have that the $\Box$-family contains eight distinct logics,
while the $\diam$-family contains four distinct logics
(something not derivable in a classical system is clearly
not derivable in the corresponding intuitionistic system).
It is also worth noticing that,  as it happens  in the classical case, 
axioms \axMbox, \axMdiam{} and \axNbox{} are interderivable, respectively, with rules
\rmbox, \rmdiam{} and \rulenbox, and that \axKbox{} is derivable from \axMbox{} and \axCbox{}
(as the standard derivations are intuitionistically valid).

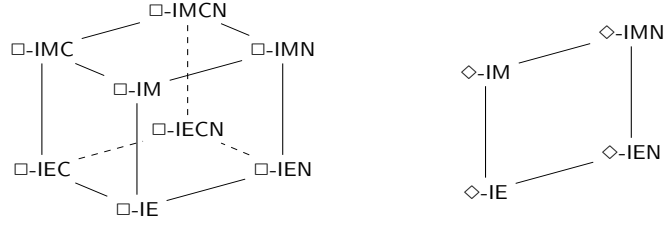
\begin{figure}
\hfill
\resizebox{9cm}{3cm}{
\begin{tikzpicture}
	\node (a) at  (0,0)  {\iboxE};
    \node (b) at (0, 2.1) {\iboxEM};
    \node  (c) at (-1.5, 0.7) {\iboxEC};
    \node (d) at (2.3, 0.7) {\iboxEN};
    \node (e) at (-1.5, 2.8) {\iboxEMC};
    \node (f) at (2.3, 2.8) {\iboxEMN};
    \node (g) at (0.8, 1.4) {\iboxECN};
    \node (h) at (0.8, 3.5) {\iboxEMCN};

\draw (a) -- (b);
\draw (a) -- (c);
\draw (a) -- (d);
\draw (b) -- (e);
\draw (b) -- (f);
\draw (c) -- (e);
\draw [dashed] (c) -- (g);
\draw (d) -- (f);
\draw [dashed] (d) -- (g);
\draw (e) -- (h);
\draw (f) -- (h);
\draw [dashed] (g) -- (h);

	\node (a2) at  (5.5,0.3)  {\idiamE};
    \node (b2) at (5.5, 2.4) {\idiamEM};
    \node (d2) at (7.8, 1) {\idiamEN};
    \node (f2) at (7.8, 3.1) {\idiamEMN};

\draw (a2) -- (b2);
\draw (a2) -- (d2);
\draw (b2) -- (f2);
\draw (d2) -- (f2);
\end{tikzpicture}
}
\hfill\qquad
\caption{\label{monomodal int cubes} The lattices of intuitionistic non-normal monomodal logics.}
\end{figure}

\subsection{Sequent calculi}

We now present sequent calculi for intuitionistic non-normal monomodal logics.
The calculi are defined as modal extensions of a given sequent calculus for \il. 
We take \gtrei{} as base calculus (Figure \ref{g3i}),
and extend it with suitable combinations of the modal rules in Figure \ref{basic modal sequent rules}.
The $\Box$-rules can be compared with the rules given 
in Lavendhomme and Lucas \cite{Lavendhomme},
where sequent calculi for classical non-normal modal logics are presented.
However, our rules are slightly different as 
(i) they have a single formula in the right-hand side of sequents;
and (ii) contexts are added to the left-hand side of sequents appearing in the conclusion.
Restriction (i) is adopted in order to have single-succedent calculi (as \gtrei{} is),
while with (ii) we implicitly embed weakening in the application of the modal rules.
We consider the sequent calculi to be defined 
by the modal rules that are added to \gtrei{}.
The calculi are the following.

\begin{figure}
\hfill
\fbox{\begin{minipage}{30em}
\vspace{0.1cm}
\begin{tabular}{p{5.5cm} p{5cm}} 

\inseq{} \ $p,\G\seq p$
&
\lbot{} \ $\bot,\G\seq A$ \\

\vspace{0.5cm} & \vspace{0.5cm}\\

\ax{$A, B, \G\seq C$}\llab{\lland}
\uinf{$A\land B, \G\seq C$}\disp 
&
 \ax{$\G\seq A$}\ax{$\G\seq B$}\llab{\rland}
\binf{$\G\seq A\land B$}\disp \\

\vspace{0.5cm} & \vspace{0.5cm}\\

\ax{$A,\G\seq C$}\ax{$B,\G\seq C$}\llab{\llor}
\binf{$A\lor B, \G\seq C$}\disp 
&
 \ax{$\G\seq A_i$}\llab{\rlor}\rlab{($i=0,1$)}
\uinf{$\G\seq A_0\lor A_1$}\disp \\

\vspace{0.5cm} & \vspace{0.5cm}\\

\ax{$A\imp B, \G\seq A$}\ax{$B, \G \seq C$}\llab{\limp}
\binf{$A\imp B, \G\seq C$}\disp 
&
\ax{$A, \G\seq B$}\llab{\rimp}
\uinf{$\G\seq A \imp B$}\disp \\
\end{tabular}
\end{minipage}}
\hfill \qquad
\caption{\label{g3i} Rules of \gtrei{} (Troelstra and Schwichtenberg \cite{Troelstra}).}
\end{figure}

\begin{figure}
\hfill
\fbox{\begin{minipage}{30em}
\vspace{0.1cm}
\begin{tabular}{p{5.5cm} p{5cm}}
\ax{$A\seq B$}\ax{$B\seq A$}\llab{\grebox}
\binf{$\G, \Box A\seq \Box B$}\disp  
&
\ax{$A\seq B$}\ax{$B\seq A$}\llab{\grediam}
\binf{$\G, \diam A\seq \diam B$}\disp  \\

\vspace{0.5cm} & \vspace{0.5cm}\\

\ax{$A\seq B$}\llab{\grmbox}
\uinf{$\G, \Box A\seq \Box B$}\disp  
&
\ax{$A\seq B$}\llab{\grmdiam}
\uinf{$\G, \diam A\seq \diam B$}\disp  \\

\vspace{0.5cm} & \vspace{0.5cm}\\

\ax{$\seq A$}\llab{\grulenbox}
\uinf{$\G\seq \Box A$}\disp 
&
\ax{$A\seq$}\llab{\grulendiam}
\uinf{$\G, \diam A\seq B$}\disp \\

\vspace{0.5cm} & \vspace{0.5cm}\\

\multicolumn{2}{l}{\ax{$A_1,...,A_n\seq B$ \quad $B\seq A_1$  ... $B\seq A_n$}\llab{\greboxc}\rlab{\ ($n\geq 1$)}
\uinf{$\G, \Box A_1, ..., \Box A_n\seq \Box B$}\disp} \\

\vspace{0.5cm} & \vspace{0.5cm}\\

\multicolumn{2}{l}{\ax{$A_1,...,A_n\seq B$}\llab{\grmboxc}\rlab{\ ($n\geq 1$)}
\uinf{$\G, \Box A_1, ..., \Box A_n\seq \Box B$}\disp} \\
\end{tabular}
\end{minipage}}
\hfill \qquad
\caption{\label{basic modal sequent rules} Modal rules for Gentzen calculi.}
\end{figure}

\vspace{0.2cm}
\begin{tabular}{l l l l l l l}
\giboxE{} & := & \grebox{} & \quad\quad & \giboxEC{} & := & \greboxc{} \\

\giboxEM{} & := & \grmbox{} & \quad & \giboxEMC{} & := & \grmboxc{} \\

\giboxEN{} & := & \grebox{} + \grulenbox & \quad & \giboxECN{} & := & \greboxc{} + \grulenbox \\

\giboxEMN{} & := & \grmbox{} + \grulenbox & \quad & \giboxEMCN{} & := & \grmboxc{} + \grulenbox \\

\gidiamE{} & := & \grediam{} \\

\gidiamEM{} & := & \grmdiam{} \\

\gidiamEN{} & := & \grediam{} + \grulendiam \\

\gidiamEMN{} & := & \grmdiam{} + \grulendiam \\
\end{tabular}

\vspace{0.2cm}
Notice that - as in Lavendhomme and Lucas \cite{Lavendhomme} - 
axiom \axCbox{} doesn't have a corresponding sequent rule,
but it is captured by modifying the rules \grebox{} and \grmbox.
In particular, these rules are replaced by \greboxc{} and \grmboxc{}, respectively, 
that are the generalisations of \grebox{} and \grmbox{} with $n$ principal formulas (instead of just one) in the left-hand side of sequents.
Observe that \greboxc{} and \grmboxc{} are non-standard, as they introduce an  
arbitrary number of modal formulas with a single application, 
and that \grebox{} has in addition an arbitrary number of premisses.
An other way to look at \greboxc{} and \grmboxc{} is to see them as infinite sets of rules, 
each set containing a standard rule for any $n\geq 1$.
Under the latter interpretation the calculi are anyway non-standard 
as they are defined by infinite sets of rules.

We now prove the admissibility of some structural rules,
and then show the equivalence between the sequent calculi and the Hilbert systems.

\begin{proposition}\label{admissibility wk ctr monomodal}
The following weakening and contraction rules are height-preserving admissible in any monomodal calculus:
\begin{center}
\ax{$\G \seq B$}
\llab{\lwk}
\uinf{$\G, A \seq B$}
\disp
\qquad
\ax{$\G \seq $}
\llab{\rwk}
\uinf{$\G \seq A$}
\disp
\qquad
\ax{$\G, A, A \seq B$}
\llab{\ctr}
\uinf{$\G, A \seq B$}
\disp.
\end{center}
\end{proposition}
\begin{proof}
By induction on $n$, we show that whenever the premiss of an application of \lwk{}, \rwk{} or \ctr{}
has a derivation of height $n$, then its conclusion has a derivation of the same height.
As usual, the proof considers the last rule applied in the derivation of the premiss
(when the premiss is not an initial sequent).
%If this is a rule of \gtrei{}, then the proof is as in Troelstra and Schwichtenberg \cite{Troelstra}. \nb{Togliere riferimento a Basic Proof Theory nella prova}
For rules of \gtrei{} the proof is standard.
For modal rules,
left and right weakening are easily handled.
For istance, the premiss $\G \seq$ of \rwk{} is necessarily derived by \grulendiam.
Then $\G$ contains a formula $\diam B$ that is 
principal in the application of \grulendiam{},
which in turn has $B\seq$ as premiss.
By a different application of \grulendiam{} to $B\seq$ 
we can derive $\G \seq A$ for any $A$.

The proof is also immediate for contraction,
where the most interesting case is possibly when both occurrences of $A$ in the premiss $\G, A, A \seq B$ of \ctr{}
are principal in the last rule applied
in its derivation.
In this case,
the last rule is either \greboxc{} or \grmboxc.
If it is \grmboxc{}, then $A\equiv \Box C$ for some $C$,
and the sequent is derived from
$D_1,...,D_n,C,C \seq$ for some $\Box D_1, ..., \Box D_n$ in $\G$.
By i.h. we can apply \ctr{} to the last sequent and obtain $D_1,...,D_n,C \seq$,
and then by \grmboxc{} derive sequent $\G, A \seq B$, which is the conclusion of \ctr{}  
(the proof is analogous for \greboxc).
\end{proof}

We now show that the  cut rule
\begin{center}
\ax{$\G \seq A$}
\ax{$\G, A \seq B$}
\llab{\cut}
\binf{$\G \seq B$}
\disp
\end{center}
is admissible in any monomodal calculus.
The proof is based on the following notion of weight of formulas:

\begin{definition}[Weight of formulas]\label{weight of formulas}
Function $\w$ assigning to each formula $A$ its weight $\w(A)$ is defined as follows:
$\w(\bot) = 0$;
$\w(p) = 1$;
$\w(A \circ B) = \w(A) + \w(B) + 1$ for $\circ \equiv \land, \lor, \imp$;
and
$\w(\Box A) = \w(\diam A) = \w(A) + 2$.
\end{definition}

Observe that,
given the present definition,    
$\neg A$ has a smaller weight than $\Box A$ and $\diam A$.
Although irrelevant to the next theorem, this will be 
used in Section \ref{section bimodal} for  the proof of cut elimination in bimodal calculi.

\begin{theorem}\label{cut elim monomodal}
Rule \cut{} is admissible in any monomodal calculus.
\end{theorem}
\begin{proof}
Given a derivation of a sequent with some applications of \cut{}, 
we show how to remove any such application and obtain 
a derivation of the same sequent without \cut.
The proof is by double induction,
with primary induction on the weight of the cut formula and subinduction on the cut height. 
We recall that, for any application of \cut,
the cut formula is the formula which is deleted by that application, while
the cut height is the sum of the heights of the derivations of the premisses of \cut.

We just consider the cases in which the cut formula is principal in the last rule applied in the derivation of both premisses of \cut.
Moreover, we treat explicitly only 
the cases in which both premisses are derived by modal rules,
as the non-modal cases are already considered in the proof of cut admissibility for \gtrei,
and because modal and non-modal rules don't interact in any relevant way.

$\bullet$
(\greboxc{}; \greboxc). \ Let $\G_1 = A_1, ..., A_n$ and $\G_2 = C_1, ..., C_m$. 
We have the following situation:

\begin{center}
\ax{$\G_1 \seq B$ \quad $B \seq A_1$ \ ... \ $B \seq A_n$}
\llab{\greboxc}
\uinf{$\G, \Box \G_1, \Box \G_2 \seq \Box B$}
\ax{$B, \G_2 \seq D$ \quad $D \seq B$ \quad $D \seq C_1$ \ ... \ $D \seq C_m$}
\rlab{\greboxc}
\uinf{$\G, \Box B, \Box\G_1,  \Box\G_2 \seq \Box D$}
\rlab{\cut}
\binf{$\G, \Box\G_1, \Box\G_2 \seq \Box D$}
\disp
\end{center}
The proof is converted as follows,
with several applications of \cut{} with $B$ as cut formula,
hence with a cut formula of smaller weight.
First we derive

\begin{center}
\ax{$\G_1 \seq B$}
\llab{\wk}
\uinf{$\G_1, \G_2 \seq B$}
\ax{$B, \G_2 \seq D$}
\rlab{\wk}
\uinf{$B, \G_1, \G_2 \seq D$}
\rlab{\cut}
\binf{$\G, \G_1, \G_2 \seq D$}
\disp
\end{center}
Then for any $1\leq i \leq n$, we derive

\begin{center}
\ax{$D \seq B$}
\ax{$B \seq A_i$}
\rlab{\wk}
\uinf{$B, D \seq A_i$}
\rlab{\cut}
\binf{$D \seq A_i$}
\disp
\end{center}
Finally we can apply \greboxc{} as follows

\begin{center}
\ax{$\G_1, \G_2 \seq D$ \quad $D \seq A_1$ \ ... \ $D \seq A_n$ \quad $D \seq C_1$ \ ... \ $D \seq C_m$}
\rlab{\greboxc}
\uinf{$\G, \Box\G_1,  \Box\G_2 \seq \Box D$}
\disp
\end{center}

$\bullet$
(\grmboxc{}; \grmboxc{}) is analogous to (\greboxc{}; \greboxc). 
(\grebox{}; \grebox{}) and (\grmbox{}; \grmbox{}) are the particular cases where $n, m = 1$.

$\bullet$
(\grulenbox{}; \greboxc). \ Let $\G_1 = B_1, ..., B_n$. The situation is as follows:

\begin{center}
\ax{$\seq A$}
\llab{\grulenbox}
\uinf{$\G, \Box \G_1 \seq \Box A$}
\ax{$A, \G_1 \seq C$ \quad $C \seq A$ \quad $C \seq B_1$ \ ... \ $C \seq B_n$}
\rlab{\greboxc}
\uinf{$\G, \Box A, \Box\G_1 \seq \Box C$}
\rlab{\cut}
\binf{$\G, \Box\G_1 \seq \Box C$}
\disp
\end{center}
The proof is converted as follows,
with an application of \cut{} on a cut formula of smaller weight.

\begin{center}
\ax{$\seq A$}
\llab{\wk}
\uinf{$\G_1 \seq A$}
\ax{$A, \G_1 \seq C$}
\llab{\cut}
\binf{$\G_1 \seq C$}
\ax{$C \seq B_1$ \ ... \ $C \seq B_n$}
\rlab{\greboxc}
\binf{$\G, \Box\G_1 \seq \Box C$}
\disp
\end{center}

$\bullet$
(\grulenbox{}; \grmboxc{}) is analogous to (\grulenbox{}; \greboxc{}). 
(\grulenbox{}; \grebox{}) and (\grulenbox{}; \grmbox{}) are the particular cases where $n = 1$.

$\bullet$ (\grediam{}; \grediam{}) and (\grmdiam{}; \grmdiam{}) 
are analogous to (\grebox{}; \grebox{}) and (\grmbox{}; \grmbox{}), respectively.

$\bullet$ (\grediam{}; \grulendiam). \ We have

\begin{center}
\ax{$A \seq B$}
\ax{$B \seq A$}
\llab{\grediam}
\binf{$\G, \diam A \seq \diam B$}
\ax{$B \seq$}
\rlab{\grulendiam}
\uinf{$\G, \diam A, \diam B \seq C$}
\rlab{\cut}
\binf{$\G, \diam A \seq C$}
\disp
\end{center}

which become

\begin{center}
\ax{$A \seq B$}
\ax{$B \seq$}
\rlab{\wk}
\uinf{$A, B \seq$}
\rlab{\cut}
\binf{$A \seq$}
\rlab{\grulendiam}
\uinf{$\G, \diam A \seq C$}
\disp
\end{center}

$\bullet$ (\grmdiam{}; \grulendiam) is analogous to (\grediam{}; \grulendiam).
\end{proof}

As a consequence of the admissibility of \cut{}
we obtain the equivalence between the sequent calculi and the axiomatic systems.

\begin{proposition}\label{equiv monomodal}
Let \X{} be any intuitionistic non-normal monomodal logic. 
Then calculus \gX{} is equivalent to system \X.
\end{proposition}
\begin{proof}
The axioms and rules of \X{} are derivable in \gX. 
For the axioms of \il{} and \modusponens{} we can consider their derivations in \gtrei,
as \gX{} enjoys admissibility of \cut.
Here we show that any modal rule allows us to derive the corresponding axiom:

\vspace{0.5cm}
\noindent
\ax{$\seq A \imp B$}
\llab{\wk}
\uinf{$A \seq A \imp B$}
\ax{$A, A\imp B \seq B$}
\llab{\cut}
\binf{$A \seq  B$}
\ax{$ \seq B\imp A$}
\rlab{\wk}
\uinf{$B \seq B\imp A$}
\ax{$B, B \imp A \seq A$}
\rlab{\cut}
\binf{$B \seq A$}
\rlab{\grebox}
\binf{$\Box A \seq \Box B$}
\rlab{\rimp}
\uinf{$\seq \Box A \imp \Box B$}
\disp

\vspace{0.5cm}
\ax{$A, B \seq A$}
\ax{$A, B \seq B$}
\llab{\rland}
\binf{$A, B \seq A \land B$}
\ax{$A, B \seq A$}
\rlab{\lland}
\uinf{$A \land B \seq A$}
\ax{$A, B \seq B$}
\rlab{\lland}
\uinf{$A \land B \seq B$}
\rlab{\greboxc}
\TrinaryInfC{$\Box A, \Box B \seq \Box(A \land B)$}
\rlab{\lland}
\uinf{$\Box A \land \Box B \seq \Box(A \land B)$}
\rlab{\rimp}
\uinf{$\seq \Box A \land \Box B \imp \Box(A \land B)$}
\disp

%\noindent
\ax{$\seq \top$}
\rlab{\grulenbox}
\uinf{$\seq \Box \top$}
\disp
\quad

%\vspace{1cm}
\noindent
\ax{$\bot \seq$}
\rlab{\grulendiam}
\uinf{$\diam\bot \seq$}
\rlab{\rneg}
\uinf{$\seq \neg \diam\bot$}
\disp
%\quad
\hfill
\ax{$A, B \seq A$}
\rlab{\lland}
\uinf{$A \land B \seq A$}
\rlab{\grmbox}
\uinf{$\Box (A \land B) \seq \Box A$}
\rlab{\rimp}
\uinf{$\seq \Box (A \land B) \imp \Box A$}
\disp
%\quad
\hfill
\ax{$A \seq A$}
\rlab{\rlor}
\uinf{$A \seq A \lor B$}
\rlab{\grmdiam}
\uinf{$\diam A \seq \diam(A \lor B)$}
\rlab{\rimp}
\uinf{$\seq \diam A \imp \diam(A \lor B)$}
\disp

\vspace{0.5cm}
\noindent
Moreover, the rules of \gX{} are derivable in \X.
As before, it suffices to consider the modal rules.
The derivations are in most cases straightforward, we just consider the following.

$\bullet$ \ If \X{} contains \axNbox, then \grulenbox{} is derivable.
Assume $\vd_\X A$. Then by \rulenbox{} (which is equivalent to \axNbox), $\vd_\X \Box A$.

$\bullet$ \ If \X{} contains \axNdiam, then \grulendiam{} is derivable.
Assume $\vd_\X A \imp \bot$. Since $\vd_\X \bot \imp A$, by \grediam, $\vd_\X \diam A \imp \diam \bot$.
Then $\vd_\X \neg\diam\bot \imp \neg\diam A$, and, since $\vd_\X \neg\diam\bot$,
we have $\vd_\X \neg\diam A$.

$\bullet$ \ If \X{} contains \axCbox, then \greboxc{} is derivable.
Assume $\vd_\X A_1\land ... \land A_n \imp B$ and $\vd_\X B\imp A_i$ for all $1\leq i\leq n$.
Then $\vd_\X B\imp  A_1\land ... \land A_n$.
By \rebox, $\vd_\X  \Box (A_1\land ... \land A_n) \imp \Box B$.
In addition, by several applications of \axCbox, $\vd_\X  \Box A_1\land ... \land \Box A_n \imp \Box (A_1\land ... \land A_n)$.
Therefore $\vd_\X  \Box A_1\land ... \land \Box A_n \imp \Box B$.
\end{proof}

\section{Intuitionistic non-normal bimodal logics}\label{section bimodal}

In this section we present intuitionistic non-normal modal logics with both %modalities 
$\Box$ and $\diam$.
In this case we first present their sequent calculi, 
and then give equivalent axiomatisations.

A simple way to define \intbilogic s would be by considering the fusion 
of two monomodal logics that belong respectively to the $\Box$- and to the $\diam$-family.
Given two logics \monomodalbox{} and \monomodaldiam{},
their fusion in language $\lbox \cup \ldiam$ is 
 the smallest bimodal logics containing 
\monomodalbox{} and \monomodaldiam{} 
 (for the sake of simplicity we can assume that
$\mathcal L_\Box$ and $\mathcal L_\diam$ share the same set of propositional variables, and
differ only with respect to $\Box$ and $\diam$).
The resulting logic is axiomatised simply by adding to \il{} the modal axioms and rules of 
\monomodalbox{}, plus the modal axioms and rules of \monomodaldiam{}.

It is clear, however, that in the resulting systems
the modalities don't interact at all, as there is no axiom involving both $\Box$ and $\diam$.
On the contrary, finding suitable interactions between the modalities is often the main issue
when intuitionistic bimodal logics are concerned.
In that case, by reflecting the fact that in \il{} connectives are not interderivable,
it is usually required that $\Box$ and $\diam$ are not dual.
We take the lacking of duality as an  additional requirement
for the definition of \intbilogic s:

\begin{itemize}
\item[(\reqthree)] $\Box$ and $\diam$ are not interdefinable.
\end{itemize}

In order to define \intbilogic s by the axiomatic systems, we would need 
to select the axioms between a plethora of possible formulas satisfying (\reqthree).
If we look for instance at the literature on intuitionistic normal modal logics,
we see that many different axioms have been considered,
and the reasons for the specific choices are varied.
We take therefore a different way, 
and define the logics starting with their sequent calculi.
In particular we proceed as follows.

\begin{figure}
\hfill 
\fbox{\begin{minipage}{29em}
\vspace{0.1cm}
\begin{tabular}{p{5cm} p{5cm}} 
\ax{$\seq A$}\ax{$B\seq $}\llab{\gintunoa}
\binf{$\G, \Box A, \diam B \seq C$}\disp 
&
\ax{$A \seq$}\ax{$\seq B$}\llab{\gintunob}
\binf{$\G, \Box A, \diam B \seq C$}\disp  
\\

\vspace{0.5cm} & \vspace{0.5cm}\\

\ax{$A,B\seq$}\ax{$\neg A\seq B$}\llab{\gintduea}
\binf{$\G, \Box A, \diam B \seq C$}\disp 
& 
\ax{$A,B\seq$}\ax{$\neg B\seq A$}\llab{\gintdueb}
\binf{$\G, \Box A, \diam B \seq C$}\disp \\

\vspace{0.5cm} & \vspace{0.5cm}\\

\multicolumn{2}{c}
{\ax{$A, B\seq$}\llab{\ginttre}
\uinf{$\G, \Box A, \diam B \seq C$}\disp} \\
\end{tabular}
\end{minipage}}
\hfill \qquad
\caption{\label{basic interaction rules} Interaction rules for sequent calculi.}
\end{figure}

\begin{itemize}
\item[(i)] 
\Intbilogic s are defined by their sequent calculi.
The calculi are conservative extensions of a given calculus for \il,
and have as modal rules some characteristic rules of \intmonologic s,
plus some rules connecting $\Box$ and $\diam$. 
In addition, we require that the \cut{} rule is admissible.
As usual, this means that adding rule \cut{} to the calculus does not extend the set of derivable sequents.

\item[(ii)] 
To the purpose of defining the basic systems,
we consider only interactions between $\Box$ and $\diam$
that can be seen as forms of ``weak duality principles''.
In order to satisfy (\reqthree), we require that these interactions are strictly weaker than \dualbox{} and \dualdiam{},
in the sense that \dualbox{} and \dualdiam{} must  not be derivable in any corresponding system.

\item[(iii)] 
We will distinguish logics that are monotonic and logics that are non-monotonic.
Moreover, the logics will be distinguished by the different strength of interactions between the modalities.
\end{itemize}

The above points are realised in practice as follows.
As before, we take \gtrei{} (Figure \ref{g3i}) as base calculus for intuitionistic logics.
This is extended with combinations of the characteristic rules of \intmonologic s in Figure \ref{basic modal sequent rules}.
The difference is that now the calculi contain both some rules for $\Box$ and some rules for $\diam$.
In order to distinguish monotonic and non-monotonic logics,
we require that the calculi contain either both \grebox{} and  \grediam{}
(in this case the corresponding logic will be non-monotonic), 
or both \grmbox{} and \grmdiam{} (corresponding to monotonic logics).
In addition, the calculi will contain some of the  interaction rules in Figure \ref{basic interaction rules}.
Since the logics are also distinguished according to the different strenghts of the interactions between the modalities, 
we require that the calculi contain either both \gintunoa{} and  \gintunob, or both \gintduea{} and \gintdueb, or \ginttre.

In the following we present the sequent calculi for \intbilogic s
obtained by following our methodology. 
After that, for each sequent calculus we present an equivalent axiomatisation.

\subsection{Sequent calculi}

In the first part, we focus on sequent calculi for logics containing only axioms between \axMbox, \axMdiam, \axNbox{} and \axNdiam{}
(that is, we don't consider axiom \axCbox{}).
The calculi are obtained by adding to \gtrei{} (Figure \ref{g3i})
suitable combinations of the modal rules in Figures \ref{basic modal sequent rules} and \ref{basic interaction rules}.
Although in principle any combination of rules could define a calculus, 
we accept only those calculi that satisfy the restrictions explained above.
This entails in particular the need of studying cut elimination.
As usual, the first step to do towards the study of cut elimination is to prove the admissibility of the other structural rules.

\begin{proposition}
Weakening and contraction are height-preserving admissible
in any sequent calculus defined by a combination of modal rules
in Figures \ref{basic modal sequent rules} and \ref{basic interaction rules}
 that satisfies the restrictions explained above.
\end{proposition}
\begin{proof}
By extending the proof of Proposition \ref{admissibility wk ctr monomodal}
with the examination of the interaction rules in Figure \ref{basic interaction rules}.
Due to their form, however, it is immediate to verify that if the premiss of \wk{} or \ctr{} is derivable by any interaction rule,
then the conclusion is derivable by the same rule.
\end{proof}

We can now examine the admissibility of \cut.
As it is stated by the following theorem, following our methodology we obtain 12 sequent calculi for 
intuitionistic non-normal bimodal logics.

\begin{theorem}\label{cut elimination without C}
We let the sequent calculi be defined by the set of modal rules 
which are added to \gtrei. The \cut{} rule is admissible in the following calculi:

\vspace{0.2cm}
\begin{tabular}{l l l}
\gunoE{} & := & \grebox{} + \grediam{} + \gintunoa{} + \gintunob \\

\gdueE{} & := & \grebox{} + \grediam{} + \gintduea{} + \gintdueb \\

\gtreE{} & := & \grebox{} + \grediam{} + \ginttre \\

\gunoM{} & := & \grmbox{} + \grmdiam{} + \ginttre \\
\end{tabular}

\vspace{0.2cm}
\noindent
Moreover, letting \g{} be any of the previous calculi, \cut{} is admissible in 

\vspace{0.2cm}
\begin{tabular}{l l l}
\g\genzndiam{} & := & \g{} + \grulendiam \\

\g\genznbox{} & := & \g{} + \grulendiam{} + \grulenbox \\
\end{tabular}
\end{theorem}
\begin{proof}
The structure of the proof is the same as the proof of Theorem \ref{cut elim monomodal}. 
Again, we consider only the cases where the cut formula is principal in the last rule applied
in the derivation of both premisses, 
with the further restriction that the last rules are modal ones.

The combinations between $\Box$-rules, or between $\diam$-rules,
have been shown in the proof of Theorem \ref{cut elim monomodal}. 
Here we consider the possible combinations of $\Box$- or $\diam$-rules with rules for interaction.

$\bullet$ (\grebox{}; \gintunoa). \ We have

\begin{center}
\ax{$A \seq B$}\ax{$B \seq A$}
\llab{\grebox}
\binf{$\G, \Box A, \diam C \seq \Box B$}
\ax{$\seq B$}\ax{$C \seq$}
\rlab{\gintunoa}
\binf{$\G, \Box A, \Box B, \diam C \seq D$}
\rlab{\cut}
\binf{$\G, \Box A, \diam C \seq D$}
\disp
\end{center}

which become

\begin{center}
\ax{$\seq B$}
\ax{$B \seq A$}
\llab{\cut}
\binf{$\seq A$}
\ax{$C \seq$}
\rlab{\gintunoa}
\binf{$\G, \Box A, \diam C \seq D$}
\disp
\end{center}

$\bullet$ (\grediam{}; \gintunoa). \ We have

\begin{center}
\ax{$A \seq B$}\ax{$B \seq A$}
\llab{\grediam}
\binf{$\G, \diam A, \Box C \seq \diam B$}
\ax{$\seq C$}\ax{$B \seq$}
\rlab{\gintunoa}
\binf{$\G, \diam A, \Box C, \diam B \seq D$}
\rlab{\cut}
\binf{$\G, \diam A, \Box C \seq D$}
\disp
\end{center}

which become

\begin{center}
\ax{$\seq C$}
\ax{$A \seq B$}
\ax{$B \seq$}
\rlab{\cut}
\binf{$A \seq$}
\rlab{\gintunoa}
\binf{$\G, \diam A, \Box C \seq D$}
\disp
\end{center}

$\bullet$ (\grebox{}; \gintunob). \ We have

\begin{center}
\ax{$A \seq B$}\ax{$B \seq A$}
\llab{\grebox}
\binf{$\G, \Box A, \diam C \seq \Box B$}
\ax{$B \seq$}\ax{$\seq C$}
\rlab{\gintunob}
\binf{$\G, \Box A, \Box B, \diam C \seq D$}
\rlab{\cut}
\binf{$\G, \Box A, \diam C \seq D$}
\disp
\end{center}

which become

\begin{center}
\ax{$A \seq B$}
\ax{$B \seq$}
\llab{\cut}
\binf{$A \seq$}
\ax{$\seq C$}
\rlab{\gintunob}
\binf{$\G, \Box A, \diam C \seq D$}
\disp
\end{center}

$\bullet$ (\grediam{}; \gintunob). \ We have

\begin{center}
\ax{$A \seq B$}\ax{$B \seq A$}
\llab{\grediam}
\binf{$\G, \diam A, \Box C \seq \diam B$}
\ax{$C \seq$}\ax{$\seq B$}
\rlab{\gintunob}
\binf{$\G, \diam A, \Box C, \diam B \seq D$}
\rlab{\cut}
\binf{$\G, \diam A, \Box C \seq D$}
\disp
\end{center}

which become

\begin{center}
\ax{$\seq C$}
\ax{$\seq B$}
\ax{$B \seq A$}
\rlab{\cut}
\binf{$\seq A$}
\rlab{\gintunob}
\binf{$\G, \diam A, \Box C \seq D$}
\disp
\end{center}

$\bullet$ (\grebox; \gintduea). \ We have:

\begin{center}
\ax{$A\seq B$}
\ax{$B\seq A$}
\llab{\grebox}
\binf{$\G, \Box A, \diam C \seq \Box B$}
\ax{$B, C \seq$}
\ax{$\neg B \seq C$}
\rlab{\gintduea}
\binf{$\G, \Box A, \Box B, \diam C \seq D$}
\rlab{\cut}
\binf{$\G, \Box A, \diam C \seq D$}
\disp
\end{center}

which is converted into the following derivation:

\begin{center}
\ax{$A\seq B$}
\llab{\wk}
\uinf{$A, C \seq B$}
\ax{$B, C\seq A$}
\rlab{\wk}
\uinf{$A, B, C \seq$}
\llab{\cut}
\binf{$A, C \seq$}
\ax{$B\seq A$}
\uinf{$\neg A \seq \neg B$}
\ax{$\neg B \seq C$}
\rlab{\wk}
\uinf{$\neg A, \neg B \seq C$}
\rlab{\cut}
\binf{$\neg A \seq C$}
\rlab{\gintduea}
\binf{$\G, \Box A, \diam C \seq D$}
\disp
\end{center}

Observe that the former derivation has two application of \cut,
both of them with a cut formula of smaller weight as, in particular,
$\w(\neg B) < \w(\Box B)$ (cf. Definition \ref{weight of formulas}).

$\bullet$ (\grebox; \ginttre) is analogous to the next case (\grmbox; \ginttre).

$\bullet$ (\grmbox; \ginttre). \ We have:

\begin{center}
\ax{$A\seq B$}
\llab{\grmbox}
\uinf{$\G, \Box A, \diam C \seq \Box B$}
\ax{$B, C \seq$}
\rlab{\ginttre}
\uinf{$\G, \Box A, \Box B, \diam C \seq D$}
\rlab{\cut}
\binf{$\G, \Box A, \diam C \seq D$}
\disp
\end{center}

which is converted into the following derivation:

\begin{center}
\ax{$A\seq B$}
\llab{\wk}
\uinf{$A, C\seq B$}
\ax{$B, C \seq$}
\rlab{\wk}
\uinf{$A, B, C\seq$}
\rlab{\cut}
\binf{$A, C\seq$}
\rlab{\inttre}
\uinf{$\G, \Box A, \diam C \seq D$}
\disp
\end{center}

$\bullet$ (\grulenbox{}; \ginttre). \ We have

\begin{center}
\ax{$\seq A$}
\llab{\grulenbox}
\uinf{$\G, \diam B \seq \Box A$}
\ax{$A, B \seq$}
\rlab{\ginttre}
\uinf{$\G, \Box A, \diam B \seq C$}
\rlab{\cut}
\binf{$\G, \diam B \seq C$}
\disp
\end{center}

which become

\begin{center}
\ax{$\seq A$}
\llab{\wk}
\uinf{$B \seq A$}
\ax{$A, B \seq$}
\rlab{\cut}
\binf{$B \seq$}
\rlab{\grulendiam}
\uinf{$\G, \diam B \seq C$}
\disp
\end{center}
\end{proof}

It is worth noticing that all cut-free calculi containing rule \grulenbox{} also contain rule \grulendiam.
In fact, combinations of rules containing \grulenbox{} and not \grulendiam{}
would give calculi where the \cut{} rule is not admissible.
This is due to the form of the interaction rules, 
that for instance
allow us to derive the sequent $\diam\bot\seq$ using \cut{} and \grulenbox.
A possible derivation is the following:
\begin{center}
\ax{$\seq \top$}
\llab{\grulenbox}
\uinf{$\diam\bot\seq \Box\top$}
\ax{$\seq \top$}
\ax{$\bot\seq$}
\rlab{\gintunoa}
\binf{$\Box\top, \diam\bot \seq$}
\rlab{\cut}
\binf{$\diam\bot \seq$}
\disp
\end{center}
Instead, sequent $\diam\bot\seq$ doesn't have any cut-free derivation
where  \grulendiam{} is not applied,
as no rule different from \grulendiam{} has $\diam\bot\seq$ in the conclusion.
We will consider in Section \ref{section CK and W} a calculus containing \grulenbox{} and not \grulendiam.
As we shall see, that calculus has interaction rules of a different form.

An additional remark concerns the possible choices of interaction rules in presence of \grmbox{} and \grmdiam.
In particular, we notice that whenever we take \grmbox{} and \grmdiam, 
rule \ginttre{} is the only interaction that gives cut-free calculi.
It can be interesting to consider a case of failure of cut elimination
when other interaction rules are considered.

\begin{example}\label{counterexample to cut elimination}
Sequent $\Box\neg p, \diam(p\land q) \seq$ is derivable from
\grmbox{} + \gintduea{} + \gintdueb{} + \cut{},
but it is not derivable from 
\grmbox{} + \gintduea{} + \gintdueb{} without \cut.
A possible derivation is as follows:

\vspace{0.3cm}
\noindent
\ax{$\neg p, p\land q \seq$}
\llab{\rneg}
\uinf{$\neg p \seq \neg(p\land q)$}
\llab{\grmbox}
\uinf{$\Box\neg p \seq \Box\neg(p\land q)$}
\llab{\wk}
\uinf{$\Box\neg p, \diam(p\land q) \seq \Box\neg(p\land q)$}
\ax{$\neg(p\land q), p\land q \seq$}
\ax{$\neg(p\land q) \seq \neg(p\land q)$}
\rlab{\gintdueb}
\binf{$\Box\neg(p\land q), \diam(p\land q) \seq$}
\rlab{\wk}
\uinf{$\Box\neg(p\land q), \Box\neg p, \diam(p\land q) \seq$}
\rlab{\cut}
\binf{$\Box\neg p, \diam(p\land q) \seq$}
\disp

\vspace{0.3cm}
Let us now try to derive bottom-up the sequent without using \cut{}.
As last rule we can only apply  \gintduea{} or \gintdueb, 
as they are the only rules with a conclusion of the right form.
In the first case the premisses would be $\neg p, p\land q\seq$, and $\neg\neg p \seq p\land q$;
while in the second case the premisses would be $\neg p, p\land q\seq$, and $\neg p \seq \neg(p\land q)$.
It is clear, however, that in both cases the second premiss is not derivable.
\end{example}

We now consider sequent calculi for logics containing axiom \axCbox.
As it happens in the case of calculi for monomodal logics, 
\axCbox{} is not captured by adding a specific rule, but we need instead to modify most of the modal rules already given.
In addition to the previous modifications of \grebox{} and \grmbox{},
we now need to modify also the interaction rules. 
In particular, we must take their generalisations 
that allow to introduce $n$ boxed formulas by a single application
(rules in Figure \ref{C interaction sequent rules}).
Rule \gintunoa{} is an exception
as the boxed formula which is principal in the rule application occurs as unboxed in the right-hand side of the premiss,
and therefore doesn't need to be changed.

As before, it can be easily shown that weakening and contraction are height-preserving admissible.

\begin{proposition}
Weakening and contraction are height-preserving admissible
in any sequent calculus defined by a combination of modal rules
in Figures \ref{basic modal sequent rules} and \ref{C interaction sequent rules}
 that satisfies our restrictions.
\end{proposition}

Following our methodology we obtain again 12 sequent calculi, as it is stated by the following theorem:

\begin{figure}
\fbox{\begin{minipage}{35em}
\vspace{0.1cm}

\ax{$A_1,...,A_n\seq$}\ax{$\seq B$}\llab{\gintunobc}
\binf{$\G, \Box A_1, ..., \Box A_n, \diam B \seq C$}\disp   
\hfill
\ax{$A_1, ..., A_n, B\seq$}\llab{\ginttrec}
\uinf{$\G, \Box A_1, ..., \Box A_n, \diam B \seq C$}\disp

\vspace{0.5cm}

\ax{$A_1, ..., A_n, B\seq$}
\ax{$\neg B\seq A_1$ \ ... \ $\neg B\seq A_n$}
\llab{\gintdueac}
\binf{$\G, \Box A_1, ..., \Box A_n, \diam B \seq C$}\disp

\vspace{0.5cm} 

\ax{$A_1, ..., A_n, B\seq$}
\ax{$\neg A_i\seq B$ \ ... \ $\neg A_n\seq B$}\llab{\gintduebc}
\binf{$\G, \Box A_1, ..., \Box A_n, \diam B \seq C$}\disp 
\end{minipage}}

\caption{\label{C interaction sequent rules} Modified interaction rules for \axCbox. 
In any rule $n\geq 1$.} 
\end{figure}

\begin{theorem}\label{cut elimination with C}
The \cut{} rule is admissible in the following calculi:

\vspace{0.2cm}
\begin{tabular}{l l l}
\gunoEC{} & := & \greboxc{} + \grediam{} + \gintunoa{} + \gintunobc \\

\gdueEC{} & := & \greboxc{} + \grediam{} + \gintdueac{} + \gintduebc \\

\gtreEC{} & := & \greboxc{} + \grediam{} + \ginttrec \\

\gunoMC{} & := & \grmboxc{} + \grmdiam{} + \ginttrec \\
\end{tabular}

\vspace{0.2cm}
\noindent
Moreover, letting \gC{} be any of the previuos calculi, \cut{} is admissible in 

\vspace{0.2cm}
\begin{tabular}{l l l}
\gC\genzndiam{} & := & \gC{} + \grulendiam \\

\gC\genznbox{} & := & \gC{} + \grulendiam{} + \grulenbox \\
\end{tabular}
\end{theorem}
\begin{proof}
As before, we only show some relevant cases.

\vspace{0.2cm}
$\bullet$ (\greboxc{}; \gintunoa). \ 
Let $\G_1$ be the multiset $A_1,...,A_n$, and $\Box\G_1$ be $\Box A_1,...,\Box A_n$.  We have:

\begin{center}
\ax{$\G_1 \seq B$ \quad $B \seq A_1$ \ ... \ $B \seq A_n$}
\llab{\greboxc}
\uinf{$\G, \Box \G_1, \diam C \seq \Box B$}
\ax{$\seq B$}
\ax{$C \seq$}
\rlab{\gintunoa}
\binf{$\G, \Box\G_1, \Box B, \diam C \seq D$}
\rlab{\cut}
\binf{$\G, \Box\G_1, \diam C \seq D$}
\disp
\end{center}

which become

\begin{center}
\ax{$\seq B$}
\ax{$B \seq A_1$}
\llab{\cut}
\binf{$\seq A_1$}
\ax{$C \seq$}
\rlab{\gintunoa}
\binf{$\G, \Box A_1, \Box A_2, ...,  \Box A_n, \diam C \seq D$}
\disp
\end{center}

\vspace{0.4cm}

$\bullet$ (\greboxc{}; \gintduebc). \ Let $\G_1= A_1,...,A_n$ and $\G_2= C_1, ..., C_k$. We have:

\begin{center}
\ax{$\G_1 \seq B$ \quad $B \seq A_1$ \ ... \ $B \seq A_n$}
\llab{\greboxc}
\uinf{$\G, \Box \G_1, \Box \G_2, \diam D \seq \Box B$}
\ax{$B, \G_2, D \seq$ \quad $\neg B \seq D$ \quad $\neg C_1 \seq D$ \ ... \ $\neg C_k \seq D$}
\rlab{\gintduebc}
\uinf{$\G, \Box\G_1, \Box B, \Box\G_2, \diam D \seq E$}
\rlab{\cut}
\binf{$\G, \Box\G_1, \Box\G_2, \diam D \seq E$}
\disp
\end{center}

\noindent
which is converted as follows:
First we derive sequent $\G_1, \G_2, D \seq$ and, for all $1\leq i \leq n$, sequent $\neg A_i \seq D$ as follows:

\vspace{0.3cm}
\ax{$\G_1\seq B$}
\llab{\wk}
\uinf{$\G_1, \G_2 \seq B$}
\ax{$B, \G_2, D \seq$}
\rlab{\wk}
\uinf{$B, \G_1, \G_2, D \seq$}
\llab{\cut}
\binf{$\G_1, \G_2, D \seq$}
\disp
\qquad
\ax{$B\seq A_i$}
\uinf{$\neg A_i \seq \neg B$}
\ax{$\neg B \seq D$}
\rlab{\wk}
\uinf{$\neg A_i, \neg B \seq D$}
\rlab{\cut}
\binf{$\neg A_i \seq D$}
\disp

\vspace{0.3cm}
\noindent
Then we can apply \gintduebc:

\begin{center}
\ax{$\G_1, \G_2, D \seq$ \quad $\neg A_1 \seq D$ \ ... \ $\neg A_n \seq D$ \quad $\neg C_1 \seq D$ \ ... \ $\neg C_k \seq D$}
\rlab{\gintduebc}
\uinf{$\G, \Box\G_1, \Box\G_2, \diam D \seq E$}
\disp
\end{center}

\vspace{0.4cm}

$\bullet$ (\greboxc{}; \ginttrec). \ Let $\G_1= A_1,...,A_n$ and $\G_2= C_1, ..., C_k$. We have:

\begin{center}
\ax{$\G_1 \seq B$ \quad $B \seq A_1$ \ ... \ $B \seq A_n$}
\llab{\greboxc}
\uinf{$\G, \Box \G_1, \Box \G_2, \diam D \seq \Box B$}
\ax{$B, \G_2, D \seq$}
\rlab{\ginttrec}
\uinf{$\G, \Box\G_1, \Box B, \Box\G_2, \diam D \seq E$}
\rlab{\cut}
\binf{$\G, \Box\G_1, \Box\G_2, \diam D \seq E$}
\disp
\end{center}

which become

\begin{center}
\ax{$\G_1 \seq B$}
\llab{\wk}
\uinf{$\G_1, \G_2, D \seq B$}
\ax{$B, \G_2, D \seq$}
\rlab{\wk}
\uinf{$\G_1, B, \G_2, D \seq$}
\rlab{\cut}
\binf{$\G_1, \G_2, D \seq$}
\rlab{\ginttrec}
\uinf{$\G, \Box\G_1, \Box\G_2, \diam D \seq E$}
\disp
\end{center}

\vspace{0.4cm}
$\bullet$ (\grmboxc{}; \ginttrec{}) is similar to the previous case.
\end{proof}

Notably, the cut-free calculi in Theorem \ref{cut elimination with C}
are the \axCbox-versions of the cut-free calculi in Theorem \ref{cut elimination with C}.
This means that, once the interaction rules are opportunely modified,
the generalisation of the rules to $n$ principal formulas
doesn't give problems with respect to cut elimination.

We have also seen that rule \gintunoa{} doesn't need to be changed.
Instead, if we don't modify the other interaction rules we obtain calculi in which \cut{} is not eliminable, 
as it is shown by the following example.

\begin{example}
Sequent $\Box p, \Box\neg p, \diam \top \seq$ is derivable by \grmboxc{} + \gintunob{} + \cut{},
but is not derivable by \grmboxc{} + \gintunob{} without  \cut{}. The derivation is as follows:
\begin{center}
\ax{$p, \neg p \seq \bot$}
\llab{\grmboxc}
\uinf{$\diam\top, \Box p, \Box\neg p \seq \Box \bot$}
\ax{$\bot\seq$}
\ax{$\seq \top$}
\rlab{\gintunob}
\binf{$\Box p, \Box\neg p, \Box\bot, \diam\top \seq$}
\rlab{\cut}
\binf{$\Box p, \Box\neg p, \diam\top \seq$}
\disp
\end{center}
Without \cut{} the sequent is instead not derivable, 
as the only applicable rule would be \gintunob,
but neither $p$ nor $\neg p$ is a contradiction.
\end{example}

\subsection{Hilbert systems}
For each sequent calculus we now
define an equivalent Hilbert system.
To this purpose, in addition the formulas in Figure \ref{modal axioms},
we also consider the axioms and rules in Figure \ref{axioms interaction}.
As before, the Hilbert systems are defined by the set of modal axioms and rules that are added to \il.
The systems are axiomatised as follows.

\begin{figure}
\fbox{\begin{minipage}{34em}
\vspace{0.1cm}
\begin{tabular}{l l l l p{0.4cm} p{2.6cm}}
\vspace{0.2cm}
\intunoa & $\neg(\Box\top\land\diam\bot)$  & \intunob & $\neg(\diam\top\land\Box\bot)$ 
&  \multirow{2}{0.4cm}{\inttre}
&  \multirow{2}{4cm}{\ax{$\neg(A\land B)$}
\uinf{$\neg(\Box A\land\diam B)$}
\disp}
\\

\intduea & $\neg(\Box A\land\diam\neg A)$\quad\quad  & \intdueb & $\neg(\Box\neg A\land\diam A)$\quad\quad \\
\end{tabular}
\end{minipage}}
\caption{\label{axioms interaction} Hilbert axioms and rules for interactions between $\Box$ and $\diam$.}
\end{figure}

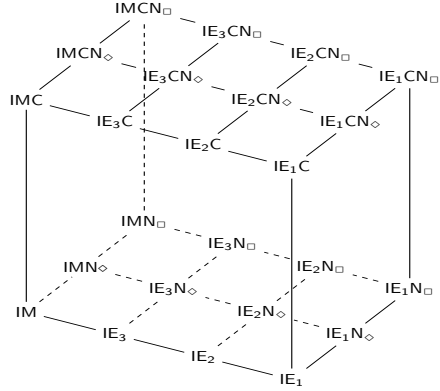
\begin{figure}
\qquad \hfill
%\resizebox{8cm}{7cm}{
\resizebox{6cm}{5.2cm}{
\begin{tikzpicture}[scale=1.1]
	\node (PW) at  (-0.5,-3)  {\unoM};
	\node (PT) at  (1,-3.5)  {\treE};
	\node (PN) at  (2.5,-4)  {\dueE};
	\node (PCL) at  (4,-4.5)  {\unoE};
	\node (PUN) at  (3.5,-3)  {\dueENdiam};
	\node (PUT) at  (2, -2.5)  {\treENdiam};
	\node (PUW) at  (0.5,-2)  {\MNdiam};
	\node (PAN) at  (4.5,-2)  {\dueENbox};
	\node (PAT) at  (3,-1.5)  {\treENbox};
	\node (PAW) at  (1.5,-1)  {\MNbox};
	\node (PU) at  (5,-3.5 ) {\unoENdiam};
	\node (PA) at  (6,-2.5 ) {\unoENbox};

	\draw (PW) -- (PT);
	\draw (PT) -- (PN);
	\draw (PN) -- (PCL);
	\draw [dashed] (PA) -- (PAN);
	\draw [dashed] (PAN) -- (PAT);
	\draw [dashed] (PAT) -- (PAW);
	\draw (PCL) -- (PU);
	\draw (PU) -- (PA);
	\draw [dashed] (PU) -- (PUN);
	\draw [dashed] (PUN) -- (PUT);
	\draw [dashed] (PUT) -- (PUW);
	\draw [dashed] (PW) -- (PUW);
	\draw [dashed] (PN) -- (PUN);
	\draw [dashed] (PUN) -- (PAN);
	\draw [dashed] (PUT) -- (PAT);
	\draw [dashed] (PT) -- (PUT);
	\draw [dashed] (PUW) -- (PAW);

	\node (V) at  (4,0.2) {\unoEC};
	\node (VA) at  (6,2.2) {\unoENboxC};
	\node (VU) at  (5,1.2) {\unoENdiamC};
	\node (VN) at  (2.5,0.7 ) {\dueEC};
	\node (VT) at  (1,1.2 ) {\treEC};	
	\node (VW) at  (-0.5,1.7 ) {\unoMC};
	\node (VUN) at  (3.5,1.7 ) {\dueENdiamC};
	\node (VUT) at  (2,2.2 ) {\treENdiamC};
	\node (VUW) at  (0.5,2.7 ) {\MNdiamC};
    \node (VAN) at  (4.5,2.7 ) {\dueENboxC};
    \node (VAT) at  (3,3.2 ) {\treENboxC};
	\node (VAW) at  (1.5,3.7 ) {\MNboxC};

	\draw (PCL) -- (V);
	\draw[dashed] (PAW) -- (VAW);
	\draw (V) -- (VU);
	\draw (PW) -- (VW);
	\draw (VU) -- (VA);
	\draw (PA) -- (VA);
	\draw (V) -- (VN);
	\draw (VN) -- (VT);
	\draw (VT) -- (VW);
	\draw (VU) -- (VUN);
	\draw (VUN) -- (VUT);
	\draw (VUT) -- (VUW);
	\draw (VN) -- (VUN);
	\draw (VT) -- (VUT);
	\draw (VW) -- (VUW);
	\draw (VA) -- (VAN);
	\draw (VAN) -- (VAT);
	\draw (VAT) -- (VAW);
	\draw (VUN) -- (VAN);
	\draw (VUT) -- (VAT);
	\draw (VUW) -- (VAW);
\end{tikzpicture}
}
\hfill \qquad
\caption{\label{int cube} The lattice of intuitionistic non-normal bimodal logics.}
\end{figure}

\vspace{0.2cm}
\begin{tabular}{l l l}
\unoE{} & := & \rebox{} + \rediam{} + \intunoa{} + \intunob \\

\dueE{} & := & \rebox{} + \rediam{} + \intduea{} + \intdueb \\

\treE{} & := & \rebox{} + \rediam{} + \inttre \\

\unoM{} & := & \rebox{} + \rediam{} + \axMbox{} + \axMdiam{} + \inttre \\
\end{tabular}

\vspace{0.2cm}
\noindent
Moreover, letting \h{} be any of the four systems listed above, we 
have the following additional systems:

\vspace{0.2cm}
\begin{tabular}{l l l}
\h\hc{} & := & \h{} + \axCbox \\

\h\hndiam{} & := & \h{} + \axNdiam \\

\h\hnbox{} & := & \h{} + \axNbox \\

\h\hc\hndiam{} & := & \h{} + \axCbox{} + \axNdiam \\

\h\hc\hnbox{} & := & \h{} + \axCbox{} + \axNbox \\
\end{tabular}

\begin{proposition}
Let \gX{} be any sequent calculus for intuitionistic non-normal bimodal logics.  
Then  \gX{} is equivalent to system \X.
\end{proposition}

\begin{proof}
Any axiom and rule of \X{} is derivable in \gX.
Here we only consider the interactions between the modalities, 
as the derivations of the other axioms have been 
already given in Proposition \ref{equiv monomodal}.

\vspace{0.5cm}
\noindent
\ax{$\seq \top$}
\ax{$\bot \seq$}
\rlab{\gintunoa}
\binf{$\Box \top, \diam \bot \seq$}
\rlab{\lland}
\uinf{$\Box \top \land \diam \bot \seq$}
\rlab{\rneg}
\uinf{$\seq \neg(\Box \top \land \diam \bot)$}
\disp
%\quad
\hfill
\ax{$\seq \top$}
\ax{$\bot \seq$}
\rlab{\gintunob}
\binf{$\diam \top, \Box \bot \seq$}
\rlab{\lland}
\uinf{$\diam \top \land \Box \bot \seq$}
\rlab{\rneg}
\uinf{$\seq \neg(\diam \top \land \Box \bot)$}
\disp
%\quad
\hfill
\ax{$\seq \neg(A\land B)$}
\llab{\wk}
\uinf{$A, B \seq \neg(A\land B)$}
\ax{$A, B, \neg(A\land B) \seq$}
\llab{\cut}
\binf{$A, B \seq$}
\rlab{\ginttre}
\uinf{$\Box A, \diam B \seq$}
\rlab{\lland}
\uinf{$\Box A \land \diam B \seq$}
\rlab{\rneg}
\uinf{$\seq \neg(\Box A \land \diam B)$}
\disp

\vspace{0.5cm}
\noindent
\ax{$A, \neg A \seq$}
\ax{$\neg A \seq \neg A$}
\rlab{\gintduea}
\binf{$\Box A, \diam \neg A \seq$}
\rlab{\lland}
\uinf{$\Box A \land \diam \neg A \seq$}
\rlab{\rneg}
\uinf{$\seq \neg(\Box A \land \diam \neg A)$}
\disp
%\quad\quad
\hfill
\ax{$A, \neg A \seq$}
\ax{$\neg A \seq \neg A$}
\rlab{\gintdueb}
\binf{$\Box \neg A, \diam A \seq$}
\rlab{\lland}
\uinf{$\Box \neg A \land \diam A \seq$}
\rlab{\rneg}
\uinf{$\seq \neg(\Box \neg A \land \diam A)$}
\disp

\vspace{0.5cm}
\noindent
Moreover, any rule of \gX{} is derivable in \X.
As before we only need to consider the interaction rules.
The derivations are immediate, we show as example the following.

$\bullet$ \ If \X{} contains axiom \intunoa, then rule \gintunoa{} is derivable.
Assume $\vd_\X A$ and $\vd_\X B \imp \bot$.
Then $\vd_\X \top\imp A$, and,
since $\vd_\X A\imp\top$, by \rebox{} we have $\vd_\X \Box A\imp \Box\top$.  
Moreover, since $\vd_\X \bot\imp B$, by \rediam{} we have $\vd_\X \diam B \imp \diam \bot$,
hence $\vd_\X \neg\diam \bot \imp \neg \diam B$.
By \intunoa{} we also have $\vd_\X \Box\top\imp\neg\diam\bot$.
Thus $\vd_\X \Box A\imp\neg\diam B$,
which gives $\vd_\X \neg(\Box A\land\diam B)$.

$\bullet$ \ If \X{} contains axiom \intdueb, then rule \gintdueb{} is derivable.
Assume $\vd_\X A \imp \neg B$ and $\vd_\X \neg B \imp A$.
Then by \rebox, $\vd_\X \Box A \imp \Box \neg B$.
By \intdueb{} we have $\vd_\X \Box\neg B \imp \neg\diam B$.
Thus $\vd_\X \Box A \imp \neg\diam B$,
which gives $\vd_\X \neg(\Box A \land\diam B)$.
\end{proof}

\section{Decidability and other consequences of cut elimination}\label{cons cut elimination}

Analytic cut-free sequent calculi are a very powerful tool for proof analysis.
In this section we take advantage of the fact that 
\cut{} is admissible in all sequent calculi defined in Sections \ref{section monomodal} and \ref{section bimodal}
in order to prove some additional properties of the corresponding logics.
By looking at the form of the rules,
we first observe that all calculi satisfy the requirements on \intlogic s that we have initially made,
\emph{i.e.} that 
they are conservative over \il{} (\reqone);
that they satisfy the disjunction property (\reqtwo);
and that the duality principles \dualbox{} and \dualdiam{} are not derivable (\reqthree).  
In a similar way we show that all calculi are pairwise non-equivalent,
hence the lattices of \intlogic s contain, respectively, 
8 distinct monomodal $\Box$-logics, 4 distinct monomodal $\diam$-logics,
and 24 distinct bimodal logics.

Some form of subformula property often follows from cut elimination.
For calculi containing rules \gintduea{} and \gintdueb{} we must consider
a property that is slightly different to the usual one,
as $\neg A$ can appear in a premiss of a rule where $\Box A$ or $\diam A$ appears in the conclusion.
As we shall see, the considered property is strong enough to provide,
together with the admissibility of contraction, 
a standard proof of decidability for \gtre{} calculi.

We conclude the section with some further remarks about the logics that we have defined,
that in particular concern the relations between intuitionistic and classical modal logics.

\begin{fact}
Any intuitionistic non-normal modal logic defined in Section \ref{section monomodal}
and Section \ref{section bimodal} satisfies requirements \reqone, \reqtwo{} and \reqthree{}
(the latter one being relevant only for bimodal logics).
\end{fact}
\begin{proof}
(\reqone) Any logic is conservative over \il{}.
In fact, the non-modal rules of each sequent calculus are exactly the rules of \gtrei.

(\reqtwo)
Any logic satisfies the disjunction property. In fact, 
given a derivable sequent of the form $\seq A\lor B$,
since no modal rule has such a conclusion, 
the last rule applied in its derivation is necessarily R$\lor$. 
This has premiss $\seq A$ or $\seq B$, which in turn is derivable.

(\reqthree)
Axioms \dualbox{} and \dualdiam{} are not derivable in \X{} for any system \X.
In particular, neither $\neg\Box\neg A \imp \diam A$,
nor $\neg\diam\neg A \imp \Box A$
(\emph{i.e.} the right-to-left implications of \dualbox{} and \dualdiam{}) is derivable.
For instance, if we try to derive bottom-up the sequent $\neg\Box\neg A \seq \diam A$ in \gX,
the only applicaple rule would be \limp.
This has premiss $\neg\Box\neg A \seq \Box\neg A$.
Again, \limp{} is the only applicable rule, with the same sequent as premiss.
Since this is not an initial sequent, we have that $\neg\Box\neg A \seq \diam A$ is not derivable.
The situation is analogous for $\neg\diam\neg A \seq \Box A$.
\end{proof}

\begin{theorem}
The lattice of intuitionistic non-normal bimodal logics contains 24 distinct systems.
\end{theorem}
\begin{proof}
We leave to the reader to check that taken two logics \logicone{} and \logictwo{} of the lattice,
we can always find some formulas (or rules) that are derivable in \logicone{} and not in \logictwo, or {\it vice versa}.
This can be easily done by considering the corresponding calculi \glogic$_1$ and \glogic$_2$.
In particular, if \logicone{} is stronger than \logictwo{}, then the characteristic axiom of \logicone{} is not derivable in \logictwo{}.
If instead \logicone{} and \logictwo{} are incomparable, then they both have some characteristic axioms (or rules)
that are not derivable in the other.
For rule \inttre{} we can consider the counterexample to cut elimination in Example \ref{counterexample to cut elimination}
\end{proof}

\begin{definition}[Strict subformula and negated subformula]
For any formulas $A$ and $B$,
we say that  $A$ is a \emph{strict subformula} of $B$ if 
$A$ is a subformula of $B$ and $A\not\equiv B$.
Moreover, 
we say that  $A$ is a \emph{negated subformula} of $B$ if 
there is a formula $C$ such that $C$ is a strict subformula of $B$ and $A\equiv \neg C$.
\end{definition}

\begin{definition}[Subformula property and negated subformula property]
We say that a sequent calculus \gX{} enjoys the \emph{subformula property} 
if all formulas in any derivation are subformulas of the endsequent.
We say that \gX{} enjoys the \emph{negated subformula property} 
if all formulas in any derivation are either subformulas or negative subformulas of the endsequent.
\end{definition}

As an immediate consequence of cut elimination we have the following result.

\begin{theorem}
Any sequent calculus different from \calcolidue{} enjoys the subformula property.
Moreover, calculi \calcolidue{} enjoy the negated subformula property.
\end{theorem}

Having that the calculi enjoy 
the above subformula properties we can easily adapt the proof of decidability for
\gtrei{} in Troelstra and Schwichtenberg \cite{Troelstra} 
and obtain thereby a proof of decidability for our calculi.

\begin{theorem}[Decidability]
For any intuitionistic non-normal modal logic defined in Section \ref{section monomodal}
and Section \ref{section bimodal}
it is decidable whether a given formula is derivable.
\end{theorem}
%\begin{proof}
%We can show that 
%the number of possible derivations of a given endsequent is finite, and each derivation is finite as well.
%For the intuitionistic fragment \gtrei{} the argument is the same as in Troelstra and Schwichtenberg \cite{Troelstra}.
%Concerning the modal rules, observe that any application bottom-up of a modal rule
%strictly reduces the complexity of the sequent to which it applies.
%\end{proof}

We conclude this section with some remarks about the logics we that have defined.
Firstly, we notice that there are three different systems
(that is \unoE{}, \dueE{}, \treE) that we can make correspond to the same classical logic 
(that is logic \E),
and the same holds for some of their extensions.
This is essentially due to the lost of duality between $\Box$ and $\diam$,
that permits us to consider interactions of different strengths
that are  equally derivable in classical logic but are not intuitionistically equivalent.
We see therefore that the picture of systems that emerge from
a certain set of logic principles
is much richer in the intuitionistic case than in the classical one.

Furthermore, logic \treE{}  (as well as its non-monotonic extensions)
leads us to the following  consideration.
It is normally expected that an intuitionistic modal logic is strictly weaker 
than the corresponding classical modal logic,
mainly because \il{} is weaker than \cl.
However, if we make correspond \treE{} to classical \E, this is not the case anymore.
In fact, rule \inttre{} is classically equivalent to \rmbox,
and hence not derivable in \E.
At the same time, however, it would be odd to consider \treE{} 
as  corresponding to classical \EM,
as neither \axMbox{} nor \axMdiam{} is derivable.

As a consequence,
this particular case suggests that assuming an intuitionistic base
not only allows us to make subtle distinctions
between principles that are not distinguishable in classical logic,
but also gives us the possibility to investigate systems that
in a sense lie between two different classical logics,
and don't correspond essentially to any of the two.

\section{Semantics}\label{section semantics}

In this section we present a semantics for all systems defined in Sections \ref{section monomodal} and \ref{section bimodal}.
As we shall see, the present semantics represents a general framework for intuitionistic modal logics,
that is able to capture modularly further intuitionistic non-normal modal logics as \CK{} and \HW. 
The models are obtained by combining intuitionistic Kripke models and neighbourhood models
(Definition \ref{classical neighbourhood models}) in the following way:

\begin{definition}\label{\intmodel} 
A \emph{\cupledintmodel} (\intmodel) is a tuple 
$\M=\langle\W,\less,\nbox,\ndiam,\V\rangle$,
where  $\W$ is a non-empty set,
$\less$ is a preorder over $\W$,
$\V$ is a hereditary valuation function $\W\longto\atm$
({\it i.e.} $w\less v$ \ implies \ $\V(w)\subseteq\V(v)$),
and $\nbox$, $\ndiam$ are two neighbourhood functions $\W\longto\pow(\pow(\W))$ 
such that:
\begin{center}
$w\less v$ \ implies \ $\nbox(w)\subseteq\nbox(v)$ and $\ndiam(w)\supseteq\ndiam(v)$ \quad ($hp$).
\end{center}
Functions $\nbox$ and $\ndiam$ can be \emph{supplemented}, \emph{closed under intersection}, or \emph{contain the unit}
(cf. properties in Definition \ref{classical neighbourhood models}).
Moreover, letting $\setneg\alpha$ denote the set 
$\{w\in\W \mid \textup{for all } v\more w, v\notin\alpha\}$,
$\nbox$ and $\ndiam$ can be related in the following ways:

\vspace{0.2cm}
\noindent
\begin{tabular}{l l l}
For all $w\in\W$, \ $\nbox(w)\subseteq\ndiam(w)$ & \Connunobis{} (\Connuno); & \\

If $\alpha\in\nbox(w)$, then $\W\setminus\setneg\alpha\in\ndiam(w)$ & \Connduebisaaux{} (\Connduea); \\

If $\setneg\alpha\in\nbox(w)$, then $\W\setminus\alpha\in\ndiam(w)$ & \Connduebisbaux{} (\Conndueb); \\

If $\alpha\in\nbox(w)$ and $\alpha\subseteq\beta$, then $\beta\in\ndiam(w)$ & \Conntrebis{} (\Conntre). \\
\end{tabular}

\vspace{0.2cm}
The forcing relation $w\Vd A$ 
associated to \intmodel s is defined inductively as follows:

\vspace{0.2cm}
\begin{tabular}{l l l}
$w\Vd p$ & iff & $p\in\V(w)$; \\
$w\not\Vd \bot$; \\
$w\Vd B\land C$ & iff & $w\Vd A$ and $w\Vd B$; \\
$w\Vd B\lor C$ & iff & $w\Vd A$ or $w\Vd B$; \\
$w\Vd B \imp C$ & iff & for all $v\more w$, $v \Vd B$ implies $v \Vd C$; \\
$w\Vd\Box B$ & iff & $[B]\in\nbox(w)$; \\
$w\Vd\diam B$ & iff & $\W\setminus[B]\notin\ndiam(w)$. \\
\end{tabular}

\vspace{0.2cm}
\noindent
\intmodel s for monomodal logics 
\monomodalbox{} and \monomodaldiam{}
are defined by removing, respectively, $\ndiam$ or $\nbox$ from the above definition
(as well as the forcing condition for the lacking modality),
and are called \boxmodel s and \diammodel s.

As usual, given a class $\mc C$ of \intmodel s, we say that
a formula $A$ is {\it satisfiable} in $\mc C$ if there are $\M\in\mc C$ and $w\in \M$ such that $w\Vd A$,
and that $A$ is {\it valid} in $\mc C$ if for all $\M\in\mc C$ and $w\in \M$,  $w\Vd A$.
\end{definition}

Observe that we are taking for $\imp$ the satisfaction clause of intuitionistic Kripke models, 
while for $\Box$ and $\diam$ we are taking  the satisfaction clauses of classical neighbourhood models.
Differently from classical neighbourhood models, however,
we have here two neighbourhood functions  $\nbox$ and $\ndiam$ 
(instead of one).
This allows us to consider different relations between the two functions
(\emph{i.e.} the interaction conditions in Definition \ref{\intmodel})
that we make correspond
to interaction axioms (and rules) with different strength.

The way functions $\nbox$ and $\ndiam$ are related to the order $\less$
by condition ($hp$) guarantees that \intmodel s 
preserve the hereditary property of intuitionistic Kripke models:

\begin{proposition}
\intmodel s satisfy the {\it hereditary} property:
for all $A\in\lan$, if $w\Vd A$ and $w\less v$, then $v\Vd A$.
\end{proposition}
\begin{proof}
By induction on $A$.
For the non-modal cases the proof is standard.
For $A\equiv \Box B, \diam B$ it is immediate by ($hp$).
\end{proof}

Depending on its axioms, to each system 
are associated models with specific properties, 
as summarised in the following table:
\begin{center}
\begin{tabular}{| l | l | l | l | l |}
\cline{1-2}\cline{4-5}
\axMbox & $\nbox$ is supplemented &\qquad\quad& \intunoa{} + \intunob{} & \Connuno \\
\cline{1-2}\cline{4-5}
\axNbox & $\nbox$ contains the unit && \intduea{} + \intdueb{} & \Conndue \\   
\cline{1-2}\cline{4-5}
\axCbox & $\nbox$ is closed under $\cap$ & & \inttre &  \Conntre{} \\
\cline{1-2}\cline{4-5}
\axMdiam & $\ndiam$ is supplemented \\
\cline{1-2}
\axNdiam & $\ndiam$ contains the unit \\
\cline{1-2}
\end{tabular}
\end{center}

\noindent
Conditions \Connduea{} and \Conndueb{} are always considered together and summarised as \Conndue.
In case of supplemented models (\emph{i.e.} when both $\nbox$ and $\ndiam$ are supplemented) %, however,  
it suffices to consider \Connuno{} as the semantic condition corresponding to any interaction axiom (or rule).
In fact, it is immediate to verify that whenever a model $\M$ is \connuno{}, and $\nbox$ or $\ndiam$ is supplemented,
then $\M$ also satisfies \Conndue{} and \Conntre.

Given the semantic properties in the above table, 
we have that \boxmodel s
coincide essentially with the neighbourhood spaces by Goldblatt \cite{Goldblatt}
(although there the property of containing the unit is not considered).
The only difference is that in Goldblatt's spaces 
the neighbourhoods are assumed to be closed with respect to the order, that is:
\begin{center}
If $\alpha\in\nbox(w)$, $v\in\alpha$ and $v\less u$, then $u\in\alpha$ \quad (\lessclosure).
\end{center}
As already observed by Goldblatt, however, this property is irrelevant from the point of view of the validity of formulas,
as a formula $A$ is valid in \boxmodel s (that are supplemented, closed under intersection, contain the unit)
if and only of it is valid in the corresponding \boxmodel s that satisfy also the \lessclosure.
It is easy to verify that the same equivalence holds if we consider \intmodel s for bimodal logics,
provided that the \lessclosure{} is demanded only for the neighbourhoods in $\nbox$,
and not for those in $\ndiam$.

It is immediate to prove soundness of  \intlogic s with respect to the corresponding \intmodel s.

\begin{theorem}[Soundness]
Any \intlogic{} is sound with respect to the corresponding \intmodel s.
\end{theorem}
\begin{proof}
It is immediate to prove that a given axiom is valid whenever the corresponding property is satisfied.
For \intduea{} and \intdueb{} notice that $\setneg[A]=[\neg A]$.
\end{proof}

We now prove completeness by the canonical model construction.
In the following, let \logic{} be any \intlogic{} and $\lan$ be the corresponding language.
We call \logic-{\it prime} any set $X$ of formulas of $\lan$ which is
consistent ($X\not\vdash_\logic\bot$),
closed under derivation ($X\vdash_\logic A$ implies $A\in X$)
and such that if $(A\lor B)\in X$, then $A\in X$ or $B\in X$.
For all $A\in \lan$, we denote with $\up A$ the class of prime sets $X$ such that $A\in X$.
The standard properties of prime sets hold, in particular:

\begin{lemma}\label{prime sets}
%$(a)$ Any consistent set is included in a prime set.
$(a)$ If $X\not\vd_\logic A\imp B$, then there is a $\logic$-prime set $Y$ such that $X\cup\{A\}\subseteq Y$ and $B\notin Y$.
$(b)$ For any $A,B\in\lan$, $\up A\subseteq\up B$ implies $\vdash_\logic A\imp B$. %\nb{Riformulare (a)?}
\end{lemma}

\begin{lemma}\label{canonical model}
Let \logic{} be any logic non containing axioms \axMbox{} and \axMdiam.
The {\it canonical model} $\Mc$ for \logic{}
is defined as
the tuple $\langle \Wc, \lessc, \nboxc, \ndiamc, \Vc \rangle$, where:

\vspace{0.1cm}
$\bullet$
$\Wc$ is the class of \logic-prime sets;

$\bullet$
for all $X,Y\in\Wc$, $X\lessc Y$ if and only if $X\subseteq Y$;

$\bullet$
$\nboxc(X) =\{\up A \mid \Box A\in X\}$;

$\bullet$
$\ndiamc(X) =\pow(\Wc) \setminus \{\Wc\setminus \up A \mid \diam A\in X\}$;

$\bullet$
$\Vc(X)=\{p\in\lan \mid p\in X\}$.

\vspace{0.1cm}
\noindent
Then
for all $X\in \Wc$ and all $A\in\lan$ we have
\begin{center}$X\Vdash A \textup{ \quad iff \quad } A\in X$.\end{center} 

\noindent
Moreover:
(i) If \logic{} contains \axNbox, then $\nboxc$ contains the unit;

(ii) If \logic{} contains \axCbox, then $\nboxc$ is closed under intersection;

(iii) If \logic{} contains \axNdiam, then $\ndiamc$ contains the unit;

(iv) If \logic{} contains \intunoa{} and \intunob, then $\Mc$ is \connuno;

(v) If \logic{} contains \intduea, then $\Mc$ is \connduea;

(vi) If \logic{} contains \intdueb, then $\Mc$ is \conndueb;

(vii) If \logic{} contains \inttre, then $\Mc$ is \conntre.
\end{lemma}
\begin{proof}
By induction on $A$ we prove that $X\Vdash A$ if and only if $A\in X$.
If $A\equiv p,\bot,B\land C,B\lor C,B\imp C$ the proof is immediate. 
If $A\equiv\Box B$:
From right to left, assume $\Box B\in X$. 
Then by definition $\up B\in \nboxc(X)$, and
by inductive hypothesis,  $\up B = [B]_{\Mc}$,
therefore $X\Vd \Box B$.
From left to right, assume $X\Vd \Box B$.
Then we have $[B]_{\Mc}\in\nboxc(X)$, and, by inductive hypothesis, $[B]_{\Mc} = \up B$.
By definition, this means that there is $C\in \lan$ such that $\Box C\in X$ and $\up C=\up B$.
Then by Lemma \ref{prime sets}, $\vd_\logic C\imp B$ and $\vd_\logic B\imp C$.
Thus by \rebox, $\vd_\logic \Box C \imp \Box B$,
and, by closure under derivation, $\Box B\in X$.
If $A\equiv\diam B$:
From right to left, assume $\diam B\in X$. 
Then by definition $\Wc\setminus\up B\notin \ndiamc(X)$, and
by inductive hypothesis, $\up B = [B]_{\Mc}$,
therefore $X\Vd \diam B$.
From left to right, assume $X\Vd \diam B$.
Then we have $\Wc\setminus[B]_{\Mc}\notin\ndiamc(X)$, and, by inductive hypothesis, $\Wc\setminus\up B\notin\ndiamc(X)$.
This means that there is $C\in \lan$ such that $\diam C\in X$ and $\up C=\up B$.
Thus, $\vd_\logic C\imp B$ and $\vd_\logic B\imp C$,
therefore by \rediam, $\vd_\logic \diam C \imp \diam B$.
By closure under derivation we then have $\diam B\in X$.

Notice also that $\Mc$ is well defined:
It follows immediately by the definition that $X\lessc Y$ implies both $\nboxc(X)\subseteq\nboxc(Y)$ and $\ndiamc(X)\supseteq\ndiamc(Y)$.

Points (i)--(vi) are proved as follows:
(i)
$\Box\top\in X$ for all  $X\in\Wc$. Then by definition $\Wc=\up\top\in\nboxc(X)$.
(ii)
Assume $\alpha,\beta\in \Nc(X)$. Then there are $A,B\in\lan$ such that
$\Box A, \Box B \in X$, $\alpha=\up A$ and $\beta=\up B$. 
By closure under derivation
we have $\Box(A\land B) \in X$, and, by definition,
$\up(A\land B)\in \nboxc(X)$, where $\up(A\land B)= \up A\cap \up B = \alpha\cap\beta$.
(iii)
$\neg\diam\bot\in X$ for all $X\in\Wc$, thus by consistency, $\diam\bot\notin X$. 
If $\Wc\setminus\up\bot\notin\nboxc(X)$, 
then there is $A\in\lan$ such that $\up A=\up\bot$ and $\diam A\in X$, that implies $\diam\bot\in X$.
Therefore $\Wc=\Wc\setminus\up\bot\in\nboxc(X)$.

(iv)
Assume by contradiction that $\alpha\in\nboxc(X)$ and $\alpha\notin\ndiamc(X)$.
Then there are $A,B\in\lan$ such that $\alpha=\up A$, $\alpha=\Wc\setminus\up B$, and $\Box A,\diam B\in X$,
therefore $\up A=\Wc\setminus\up B$.
By the properties of prime sets, this implies $\vd_\logic \neg(A\land B)$ and $\vd_\logic A\lor B$,
and by the disjunction property, $\vd_\logic A$ or $\vd_\logic B$.
If we assume $\vd_\logic A$, then $\vd_\logic A\imp\coimp\top$ and $\vd_\logic B\imp\coimp\bot$.
Therefore by \rebox{} and \rediam, $\vd_\logic \Box A\imp\Box\top$ and $\vd_\logic \diam B\imp\diam\bot$,
thus by closure under derivation, $\Box\top, \diam\bot \in X$.
But $\neg(\Box\top\land\diam\bot)\in X$,
against the consistency of prime sets.
If we now assume $\vd_\logic B$, then $\vd_\logic B\imp\coimp\top$ and $\vd_\logic A\imp\coimp\bot$.
We obtain an analogous contradiction by 
$\neg(\diam\top\land\Box\bot)$.

(v) 
Assume $\alpha\in\nboxc(X)$.
Then there is $A\in\lan$ such that $\alpha=\up A$ and $\Box A\in X$.
Thus, by \intduea{} and consistency of $X$, $\diam\neg A\notin X$.
Therefore $\Wc\setminus\up\neg A \in\ndiamc(X)$
(otherwise there would be $B\in\lan$ such that $\up B=\up\neg A$ and $\diam B\in X$,
which implies $\diam\neg A\in X$).
Since $\up\neg A=\setneg\up A$
($\up\neg A=[\neg A]_{\Mc}=\setneg [A]_{\Mc}=\setneg\up A$)
and $\setneg\up A=\setneg\alpha$, we have the claim.

(vi)
By contraposition, assume $\Wc\setminus\alpha\notin\ndiamc(X)$.
Then there is $A\in\lan$ such that $\Wc\setminus\alpha=\Wc\setminus\up A$ and $\diam A\in X$.
Thus $\alpha=\up A$, and by \intdueb, $\Box\neg A\notin X$.
Therefore $\up\neg A\notin\nboxc(X)$
(otherwise there would be $\Box B\in X$ such that $\up\neg A=\up B$, 
which implies $\Box\neg A\in X$).
Since $\up\neg A=\setneg\up A=\setneg\alpha$, we have the claim.

(vii)
Assume by contradiction that $\alpha\in\nboxc(X)$, $\alpha\subseteq\beta$, and $\beta\notin\ndiamc(X)$.
Then there are $A,B\in\lan$ such that $\alpha=\up A$, $\beta=\Wc\setminus\up B$ and $\Box A,\diam B \in X$.
Moreover, $\up A\subseteq\Wc\setminus\up B$,
which implies $\up A\cap\up B=\emptyset$.
Thus $\vd_\logic\neg(A\land B)$; 
and by \inttre{} we have $\vd_\logic\neg(\Box A\land\diam B)$, against the consistency of $X$.
\end{proof}

For logics containing \axMbox{} or \axMdiam{} we slightly change the definition of canonical model.
We shorten the proof by considering, instead of \axMbox{} and \axMdiam, the syntactically equivalent rules 
\rmbox{} and \rmdiam.

\begin{lemma}\label{supplemented canonical model}
Let \logic{} be any logic containing axioms \axMbox{} and \axMdiam.
The {\it canonical model} $\Mcplus$ for $\logic$ is the tuple $\langle\Wc,\lessc,\nbox^+,\ndiam^+,\Vc\rangle$,
where $\Wc, \lessc, \Vc$ are defined as in Lemma \ref{canonical model}, and: 

\vspace{0.1cm}
$\nboxplus(X) =\{\alpha\subseteq \Wc \mid \textup{there is } A\in \lan \textup{ s.t. }  \Box A\in X \textup{ and} \up A\subseteq \alpha\}$;

$\ndiamplus(X) =\pow(\Wc) \setminus \{\alpha\subseteq \Wc  \mid 
\textup{there is } A\in \lan \textup{ s.t. }  \diam A\in X  \textup{ and } \alpha\subseteq\Wc\setminus \up A\}$.

\vspace{0.1cm}
\noindent
Then we have that 
$X\Vdash A$ if and only if $A\in X$. 
Moreover,
points (i)--(iii) of Lemma \ref{canonical model} still hold.
Finally, (iv) if \logic{} contains \inttre, then $\Mcplus$ is \connuno.
\end{lemma}
\begin{proof}
It is immediate to verify that both $\nboxplus$ and $\ndiamplus$ are supplemented.
As before, the proof is by induction on $A$. We only show the modal cases.
If $A\equiv\Box B$:
From right to left, assume $\Box B\in X$. 
Then by definition $\up B\in \nboxplus(X)$, and
by inductive hypothesis, $\up B = [B]_{\Mcplus}$,
therefore $X\Vd \Box B$.
From left to right, assume $X\Vd \Box B$.
Then we have $[B]_{\Mcplus}\in\nboxplus(X)$, and, by inductive hypothesis, $[B]_{\Mcplus} = \up B$.
By definition, this means that there is $C\in \lan$ such that $\Box C\in X$ and $\up C\subseteq\up B$, 
which implies $\vd_\logic C\imp B$.
Thus by \rmbox, $\vd_\logic \Box C \imp \Box B$,
and, by closure under derivation, $\Box B\in X$.
If $A\equiv\diam B$:
From right to left, assume $\diam B\in X$. 
Then by definition $\Wc\setminus\up B\notin \ndiamplus(X)$, and
by inductive hypothesis, $\up B = [B]_{\Mcplus}$,
therefore $X\Vd \diam B$.
From left to right, assume $X\Vd \diam B$.
Then we have $\Wc\setminus[B]_{\Mc}\notin\ndiamplus(X)$, and, by inductive hypothesis, $\Wc\setminus\up B\notin\ndiamplus(X)$.
This means that there is $C\in \lan$ such that $\diam C\in X$ and 
$\Wc\setminus B \subseteq \Wc\setminus C$,
that is $\up C\subseteq \up B$.
Thus, $\vd_\logic C\imp B$,
therefore by \rediam, $\vd_\logic \diam C \imp \diam B$.
By closure under derivation we then have $\diam B\in X$.

Points (i)--(iii) are very similar to points (i)--(iii) in Lemma \ref{canonical model}.
(iv) By contradiction, assume $\alpha\in\nboxplus(X)$ and $\alpha\notin\ndiamplus(X)$.
Then there are $A,B\in\lan$ such that $\up A\subseteq\alpha$, $\alpha\subseteq\Wc\setminus\up B$, and $\Box A, \diam B\in X$.
Therefore $\up A\subseteq\Wc\setminus\up B$, which implies $\vd_{\textup{\logic}} \neg(A\land B)$.
By \inttre{} we then have $\neg(\Box A \land \diam B)\in X$, against the consistency of $X$.
\end{proof}

\begin{theorem}[Completeness]\label{theorem completeness}
Any intuitionistic non-normal bimodal logic is complete with respect to the corresponding \intmodel s.
\end{theorem}
\begin{proof}
Assume $\not\vd_\logic A$. Then $\not\vd_\logic\top\imp A$, 
thus, by Lemma \ref{prime sets}, there is a $\logic$-prime set $\Pi$ such that $A\notin\Pi$.
By definition, $\Pi\in\Mc_{(+)}$, and by Lemma \ref{canonical model}, $\Mc_{(+)},\Pi\not\Vd A$.
By the properties of $\Mc_{(+)}$ we obtain completeness with respect to the corresponding models.
\end{proof}

It is immediate to verify that
by removing $\ndiamc$ ($\ndiamplus$) or $\nboxc$ ($\nboxplus$) from the definition of $\Mc$ ($\Mcplus$),
we obtain analogous results for monomodal logics.

\begin{theorem}
Any intuitionistic non-normal monomodal logic is complete with respect to the corresponding \intmodel s.
\end{theorem}

\subsection{Finite model property and decidability}\label{subsection filtrations}
We have seen that all \intlogic s defined in Section \ref{section monomodal} and \ref{section bimodal}
are sound and complete with respect to  a certain class of \intmodel s.
By applying the technique of filtrations to this kind of models,
we now show that most of them have also the finite model property,
thus providing an alternative proof of decidability.
The proofs are given explicitly for bimodal logics,
while the simpler proofs for monomodal logics can be easily extracted.

Given a \intmodel{} $\M$  and a set $\Phi$ of formulas of $\lan$ 
that is closed under subformulas,
we define the equivalence relation $\sim$ on $\W$ as follows:
\begin{center}
$w\sim v$ \quad iff \quad for all $A\in\Phi$, \ $w\Vd A$ iff $v\Vd A$.
\end{center}
For any $w\in\W$ and $\alpha\subseteq\W$, 
we denote with $\wclass$ the equivalence class containing $w$,
and with $\alphaclass$ the set $\{\wclass \mid w\in\alpha\}$
(thus in particular $\Aclass_\M$ is the set $\{\wclass \mid w\in[A]_\M\}$).

\begin{definition}\label{filtration}
Let $\M=\langle\W,\less,\nbox,\ndiam,\V\rangle$ be any \intmodel{}
and $\Phi$ be a set of formulas of $\lan$ closed under subformulas.
A {\it filtration} of $\M$ through $\Phi$ (or $\Phi$-filtration)
is any model $\Mf=\langle\Wf,\lessf,\nboxf,\ndiamf,\Vf\rangle$ such that:

\begin{itemize}
\item $\Wf = \{\wclass \mid w\in\W\}$;

\item $\wclass \lessf \vclass$ \ iff \ for all $A\in\Phi$, $\M,w\Vd A$ implies $\M,v\Vd A$;

\item for all $\Box A \in\Phi$, \quad $\Aclass_\M\in\nboxf(\wclass)$ \ iff \ $[A]_\M\in\nbox(w)$;

\item for all $\diam A \in\Phi$, \quad $\Wf\setminus\Aclass_\M\in\ndiamf(\wclass)$ \ iff \  $\W\setminus[A]_\M\in\ndiam(w)$;

\item for all $p\in\Phi$, \quad $p\in\Vf(\wclass)$ \ iff \ $p\in\V(w)$.
\end{itemize}
\end{definition}

Observe that model $\Mf$ is well-defined, as
for all $\Box A, \diam B, p\in\Phi$ we have that $w\sim v$ implies
(i) $\Aclass_\M\in\nboxf(\wclass)$ iff $\Aclass_\M\in\nboxf(\vclass)$;
(ii) $\Wf\setminus\Bclass_\M\in\ndiamf(\wclass)$ iff $\Wf\setminus\Bclass_\M\in\ndiamf(\vclass)$; and
(iii) $p\in\Vf(\wclass)$ iff $p\in\Vf(\vclass)$.
Moreover, it is immediate to verify that
(iv) $\lessf$ is a preorder;
(v) $\Vf$ is hereditary;
(vi) if $\wclass\lessf\vclass$ and $\Box A\in\Phi$, then $\Aclass_\M\in\nboxf(\wclass)$ implies $\Aclass_\M\in\nboxf(\vclass)$;
(vii) if $\wclass\lessf\vclass$ and $\diam B\in\Phi$, then $\Wf\setminus\Bclass_\M\in\ndiamf(\vclass)$ implies $\Wf\setminus\Bclass_\M\in\ndiamf(\wclass)$;
and
(viii) for all $\alpha\subseteq\Wf$, $\alpha\in\nboxf(\wclass)$ implies $\alpha\in\ndiamf(\wclass)$.
Thus $\Mf$ is a \intmodel.

\begin{lemma}[Filtrations lemma]\label{filtration lemma}
For any formula $A\in\Phi$, 
\begin{center}
$\Mf,\wclass \Vd A \textup{ \ iff \ } \M,w\Vd A$.
\end{center}
\end{lemma}
\begin{proof}
Notice that this is equivalent to prove that $[A]_{\Mf}=\Aclass_\M$.
The proof is by induction on $A$.
For $A \equiv p, \bot, B\land C, B\lor C$ the proof is immediate.

$A \equiv B\imp C$. 
Assume $\M,w\not\Vd B\imp C$. Then there is $v\more w$ such that $\M,v\Vd B$ and $\M,v\not\Vd C$.
By inductive hypothesis $\Mf,\vclass\Vd B$ and $\Mf,\vclass\not\Vd C$.
Moreover, by definition of $\lessf$ (and monotonicity of $\M$), $\wclass\lessf\vclass$.
Therefore $\Mf,\wclass\not\Vd B \imp C$.
Now assume $\Mf,\wclass\not\Vd B \imp C$.
Then there is $\vclass\in\Wf$ such that $\wclass\lessf\vclass$, $\Mf,\vclass\Vd B$ and $\Mf,\vclass\not\Vd C$.
By inductive hypothesis $\M,v\Vd B$ and $\M,v\not\Vd C$, thus $\M,v\not\Vd B\imp C$.
By definition of $\lessf$ we then have $\M,w\not\Vd B\imp C$.

$A\equiv \Box B$.
$\Mf,\wclass\Vd\Box B$ iff
$[B]_{\Mf}\in\nboxf(\wclass)$ iff (i.h.)
$\Bclass_\M\in\nboxf(\wclass)$ iff
$[B]_{\M}\in\nbox(w)$ iff 
$\M,w\Vd\Box B$.

$A\equiv \diam B$.
$\Mf,\wclass\Vd\diam B$ iff
$\Wf\setminus[B]_{\Mf}\notin\ndiamf(\wclass)$ iff (i.h.)
$\Wf\setminus\Bclass_\M\notin\ndiamf(\wclass)$ iff
$\W\setminus[B]_{\M}\notin\ndiam(w)$ iff 
$\M,w\Vd\diam B$.
\end{proof}

\begin{lemma}
Let $\Mf$ be a $\Phi$-filtration of $\M$.
$(i)$ If $\nbox(w)$ contains the unit and $\Box\top\in\Phi$, then $\nboxf(\wclass)$ contains the unit.
$(ii)$ If $\ndiam(w)$ contains the unit and $\diam\bot\in\Phi$, then $\ndiamf(\wclass)$ contains the unit.
%\nb{Messo volutamente $\diam\bot$ e non $\neg\diam\bot$: qui non mi interessa l'assioma ma che ci sia la formula $\diam\bot$.}
\end{lemma}
\begin{proof}
Immediate by Definition \ref{filtration} and Lemma \ref{filtration lemma}.
\end{proof}

\begin{definition}
We call {\it finest} $\Phi$-filtration (cf. Chellas \cite{Chellas}) any $\Phi$-filtration $\Mf$ of $\M$ such that:

\vspace{0.2cm}
\begin{tabular}{l}
$\nboxf(\wclass)=\{\Aclass_\M \mid \Box A\in\Phi \textup{ and } [A]_\M\in\nbox(w)\}$; and \\

$\ndiamf(\wclass)=\mc P(\Wf) \setminus \{\Wf\setminus\Aclass_\M \mid 
\diam A\in\Phi  \textup{ and } \W\setminus[A]_\M\notin\ndiam(w) \}$. \\
\end{tabular}

\vspace{0.2cm}
\noindent
Moreover, let $\M^\circ=\langle \Wf,\lessf,\nbox^\circ,\ndiam^\circ,\Vf\rangle$ be a \intmodel{}
where $\Wf$, $\lessf$ and $\Vf$ are as in $\Mf$.
We say that:

\vspace{0.1cm}
$\bullet$ \  $\M^\circ$ is the \emph{supplementation} of $\Mf$ if: 

\vspace{0.1cm}
\begin{tabular}{ll}
& $\alpha\in\nbox^\circ(\wclass)$ \ iff \  there is $\beta\in\nboxf(\wclass)$ s.t. $\beta\subseteq\alpha$; \\
& $\alpha\notin\ndiam^\circ(\wclass)$ \ iff \ there is $\beta\notin\ndiamf(\wclass)$ s.t. $\alpha\subseteq\beta$. \\
\end{tabular}

\vspace{0.1cm}
$\bullet$ \ $\M^\circ$ is the \emph{intersection closure} of $\Mf$ if 
$\ndiam^\circ(\wclass) = \ndiamf(\wclass)$, and

\vspace{0.1cm}
\begin{tabular}{ll}
& $\alpha\in\nbox^\circ(\wclass)$ \ iff \  there are $\alpha_1,...,\alpha_n\in\nboxf(\wclass)$ s.t. $\alpha_1 \cap ... \cap \alpha_n=\alpha$.
\end{tabular}

\vspace{0.1cm}
$\bullet$ \  $\M^\circ$ is the \emph{quasi-filtering} of $\Mf$ if: 

\vspace{0.1cm}
\begin{tabular}{ll}
& $\alpha\in\nbox^\circ(\wclass)$ \ iff \  there are $\alpha_1,...,\alpha_n\in\nboxf(\wclass)$ s.t. $\alpha_1 \cap ... \cap \alpha_n\subseteq\alpha$; \\
& $\alpha\notin\ndiam^\circ(\wclass)$ \ iff \ there is $\beta\notin\ndiamf(\wclass)$ s.t. $\alpha\subseteq\beta$. \\
\end{tabular}
\end{definition}

It is immediate to verify that
the supplementation of a model $\M$ is supplemented,
its intersection closure is closed under intersection,
and its quasi-filtering is both supplemented and closed under intersection.

\begin{lemma}
Let $\Mf$ be a finest $\Phi$-filtration of $\M$.
$(i)$ If $\M$ is \connuno, then $\Mf$ is \connuno.
$(ii)$ If $\M$ is \conntre, then $\Mf$ is \conntre.
\end{lemma}
\begin{proof}
$(i)$ Assume by contradiction that $\alpha\in\nboxf(\wclass)$ and $\alpha\notin\ndiamf(\wclass)$.
Then $\alpha=\Aclass_\M$ for a $A\in\lan$ such that $\Box A\in\Phi$ and $[A]_\M\in\nbox(w)$.
Moreover $\alpha=\Wf\setminus\Bclass_\M$ for a $B\in\lan$ such that $\diam B\in\Phi$ and $\W\setminus[B]_\M\notin\ndiam(w)$.
Thus $\Aclass_\M = \Wf\setminus\Bclass_\M$, which implies $[A]_\M=\W\setminus[B]_\M$
($w\in[A]_\M$ iff $\wclass\in\Aclass_\M$ iff $\wclass\in\Wf\setminus\Bclass_\M$ iff $w\in\W\setminus[B]_\M$).
Then, since $\M$ is \connuno, $\W\setminus[B]_\M\in\ndiam(w)$,
which gives a contradiction.

$(ii)$ 
Assume by contradiction that $\alpha\in\nboxf(\wclass)$, $\alpha\subseteq\beta$ and $\beta\notin\ndiamf(\wclass)$.
Then $\alpha=\Aclass_\M$ for a $A\in\lan$ such that $\Box A\in\Phi$ and $[A]_\M\in\nbox(w)$.
Moreover $\beta=\Wf\setminus\Bclass_\M$ for a $B\in\lan$ such that $\diam B\in\Phi$ and $\W\setminus[B]_\M\notin\ndiam(w)$.
Thus $\Aclass_\M \subseteq \Wf\setminus\Bclass_\M$, which implies $[A]_\M \subseteq \W\setminus[B]_\M$.
Then, since $\M$ is \conntre, $\W\setminus[B]_\M\in\ndiam(w)$, 
which gives a contradiction.
\end{proof}

\begin{lemma}
Let $\M$, $\Mf$ and $\Mcirc$ be \intmodel s,
where $\Mf$ is a finest $\Phi$-filtration of $\M$ for a set $\Phi$ of formulas
that is closed under subformulas. 
We have:

\begin{itemize}
\item[(i)] If $\M$ is supplemented and \connuno,
and $\Mcirc$ is the supplementation of $\Mf$, 
then $\Mcirc$ is \connuno{} and is a $\Phi$-filtration of $\M$.

\item[(ii)] If $\M$ is closed under intersection and \connuno,
and $\Mcirc$ is the closure under intersection of $\Mf$, 
then $\Mcirc$ is \connuno{} and is a $\Phi$-filtration of $\M$.

\item[(iii)] If $\M$ is supplemented, closed under intersection, and \connuno,
and $\Mcirc$ is the quasi-filtering of $\Mf$, 
then $\Mcirc$ is \connuno{} and is a $\Phi$-filtration of $\M$.

\item[(iv)] If $\M$ is closed under intersection and \conntre,
and $\Mcirc$ is the closure under intersection of $\Mf$, 
then $\Mcirc$ is \conntre{} and is a $\Phi$-filtration of $\M$.
\end{itemize}
\end{lemma}
\begin{proof}
Points (i)--(iv) are proved similarly. 
We show as example the proof of point (iii).
Firstly we prove by contradiction that $\Mcirc$ is \connuno.
Assume $\alpha\in\nboxcirc(\wclass)$ and $\alpha\notin\ndiamcirc(\wclass)$.
Then there are $\alpha_1, ..., \alpha_n \in\nboxf(\wclass)$ s.t. $\alpha_1\cap ...\cap \alpha_n\subseteq\alpha$;
and there is $\beta\notin\ndiamf(\wclass)$ s.t. $\alpha\subseteq\beta$.
By definition, this means that there are $\Box A_1, ..., \Box A_n\in\Phi$ s.t.
$\alpha_1=\Aunoclass_{\M}$, ..., $\alpha_n=\Anclass_{\M}$, and $[A_1]_{\M}, ..., [A_n]_{\M}\in\nbox(w)$.
Moreover, there is $\diam B\in\Phi$ s.t. $\beta=\Wf\setminus\Bclass_{\M}$ and $\W\setminus[B]_{\M} \notin\ndiam(w)$.
As a consequence, we also have $\Aunoclass_{\M}\cap ... \cap \Anclass_{\M}\subseteq \Wf\setminus\Bclass_{\M}$.
Since $\Mf$ is a $\Phi$-filtration of $\M$, by the filtration lemma this implies 
$[A_1]_{\M}\cap ...\cap [A_n]_{\M}\subseteq \W\setminus[B]_{\M}$.
Then by intersection closure of $\nbox$, $[A_1]_{\M}\cap ...\cap [A_n]_{\M}\in\nbox(w)$,
and by its supplementation, $\W\setminus[B]_{\M}\in\nbox(w)$.
Finally, since $\M$ is \connuno, $\W\setminus[B]_{\M}\in\ndiam(w)$,
which gives a contradiction.

We now prove that $\Mcirc$ is a $\Phi$-filtration of $\M$.
Let $\Box A\in \Phi$. 
If $[A]_{\M}\in\nbox(w)$, then $\Aclass_{\M}\in\nboxf(\wclass)$, and also $\Aclass_{\M}\in\nboxcirc(\wclass)$.
Now assume $\Aclass_{\M}\in\nboxcirc(\wclass)$.
Then there are $\alpha_1, ..., \alpha_n \in\nboxf(\wclass)$ s.t. $\alpha_1\cap ...\cap \alpha_n\subseteq \Aclass_{\M}$.
By definition, this means that there are $\Box A_1, ..., \Box A_n\in\Phi$ s.t.
$\alpha_1=\Aunoclass_{\M}$, ..., $\alpha_n=\Anclass_{\M}$, and $[A_1]_{\M}, ..., [A_n]_{\M}\in\nbox(w)$.
Then, since $\Mf$ is a $\Phi$-filtration of $\M$, $[A_1]_{\M}\cap ...\cap [A_n]_{\M}\subseteq [A]_{\M}$.
By intersection closure of $\nbox$, $[A_1]_{\M}\cap ...\cap [A_n]_{\M}\in\nbox(w)$,
then by supplementation, $[A]_{\M}\in\nbox(w)$.

Now let $\diam A\in\Phi$.
If $\W\setminus[A]_{\M}\notin\ndiam(w)$, then $\Wf\setminus\Aclass_{\M}\notin\ndiamf(\wclass)$, 
and also $\Wf\setminus\Aclass_{\M}\notin\ndiamcirc(\wclass)$.
Now assume $\Wf\setminus\Aclass_{\M}\notin\ndiamcirc(\wclass)$.
Then there is $\beta\notin\ndiamf(\wclass)$ s.t. $\Wf\setminus\Aclass_{\M}\subseteq\beta$.
By definition, $\beta=\Wf\setminus\Bclass_{\M}$ for a $\diam B\in\Phi$ s.t. $\W\setminus[B]_{\M}\notin\ndiam(w)$.
Since $\Mf$ is a $\Phi$-filtration of $\M$, we have  $\W\setminus[A]_{\M}\subseteq\W\setminus[B]_{\M}$.
Then by supplementation, $\W\setminus[A]_{\M}\notin\ndiam(w)$.
\end{proof}

\begin{theorem}
If a formula $A$ is satisfiable in a \intmodel{} $\M$ that is \connuno{} or \conntre{},
then $A$ is satisfiable in a \intmodel{} $\M'$ 
with the same properties of $\M$
and in addition is finite.
\end{theorem}
\begin{proof}
Standard, by taking $\Phi = \sbf(A) \cup \{\Box\top, \diam\bot, \top, \bot\}$ 
and, depending on the properties of $\M$, the right transformation $\M'$ of $\M$.
Observe that whenever $\Phi$ is finite,
any $\Phi$-filtration $\M'$ of $\M$ is finite as well.
\end{proof}

\begin{corollary}
Any \intbilogic{} different from \logichedue{} enjoys the finite model property.
Moreover, any \intmonologic{} enjoys the finite model property.
\end{corollary}

\section{Constructive \textbf{\textsf K} and propositional \CCDL}\label{section CK and W}

We have seen in Section \ref{section semantics} that  
\boxmodel s coincide essentially with Goldblatt's neighbourhood spaces.
In Fairtlough and Mendler \cite{Fairtlough}, Goldblatt's spaces 
are considered in order to provide a semantics for Propositional Lax Logic (\PLL{}),
an intuitionistic monomodal logic for hardware verification 
that fails to validate the rule of necessitation.

We show in this section that the framework of \intmodel s
is also adapted to cover two additional  well-studied 
\intbilogic s, namely  \CK{} (for ``constructive \K'') by Bellin \emph{et al.}~\cite{Bellin}, and 
 \HW{}, as we call the propositional fragment of \wij's first-order logic \CCDL{} (\wij{} \cite{\wij}).
In particular, we show that the two systems 
can be included in our framework
by considering a  very simple additional property.

Different possible worlds semantics have already been given for the two logics.
In particular, logic \HW{} has both a relational semantics (\wij{} \cite{\wij}) 
and a neighbourhood semantics (Kojima \cite{Kojima}),
while a relational semantics for \CK{} 
has been given in Mendler and de Paiva \cite{Mendler1} by adding inconsistent worlds to the relational models for \HW{}.
However, if compared to the existing ones, our semantics has the advantage
of including \CK{} and \HW{} in a more general framework, 
that shows how  the two systems can be obtained as extensions of weaker logics in a modular way.
In addition, two further benefits concern specifically \CK{}.
In particular, to the best of our knowledge we are presenting the first neighbourhood semantics for this system.
Moreover, and most notably, this kind of models don't make use of inconsistent worlds.

In the following we first present logics \CK{} and \HW{}
by giving both the axiomatisations and the sequent calculi.
After that we define their \intmodel s and prove soundness and completeness.
Finally, we present their pre-existing possible worlds semantics 
and prove directly their equivalence with \intmodel s.

\subsection{Hilbert systems and sequent calculi}

Logic \CK{} (Bellin \emph{et al.}~\cite{Bellin}) is Hilbert-style defined by adding to \il{} the following axioms and rules:
\begin{center}
\axKbox{} \ $\Box(A\imp B)\imp (\Box A \imp \Box B)$,   \hfill 
\axKdiam{} \ $\Box(A\imp B)\imp (\diam A \imp \diam B)$, 
\hfill
\ax{$A$}\llab{\rulenbox}\uinf{$\Box A$}\disp.
\end{center}
Logic \HW{} is the extension of \CK{} with axiom  \axNdiam{} ($\neg\diam\bot$).%
\footnote{The axiomatisation given by \wij{} \cite{\wij} includes also $\diam(A \imp B) \imp (\Box A \imp \diam B)$; 
however this formula is derivable from the other axioms (cf.~e.g.~Simpson \cite{Simpson}, p.~48).}
It is worth noticing that,
given the syntactical equivalences that we have recalled in Section \ref{section monomodal},
an equivalent axiomatisation for \CK{} is obtained
by extending \il{} with
rules \rebox{} and \rediam, and axioms \axMbox, \axNbox, \axCbox, and \axKdiam{}
(as before, by adding also \axNdiam{} we obtain logic \HW{};
notice that axiom \axMdiam{} is derivable in both systems, 
{\it e.g.} from \rulenbox{} and \axKdiam).

Logics \CK{} and \HW{} are non-normal as they reject some form of distributivity of $\diam$ over $\lor$.
In particular, \HW{} rejects binary distributivity (\axCdiam), while \CK{} rejects both binary and nullary distributivity (\axCdiam, \axNdiam).
The modality $\Box$ is instead normal as the systems contain axiom \axKbox{} and the rule of necessitation.

Sequent calculi for \CK{} and \HW{} 
(denoted here as \GCK{} and \GW)
are defined, respectively, in Bellin \emph{et al.}~\cite{Bellin} and \wij{} \cite{\wij}.
In order to present the calculi we consider the following rule,
that we call \Wrule{} (for ``\wij''):

\begin{center}
\ax{$A_1, ...,A_n,B \seq C$}
\llab{\Wrule}
\rlab{\quad  $(n\geq 1)$} 
\uinf{$\G,\Box A_1, ...,\Box A_n,\diam B \seq \diam C$}
\disp.
\end{center}

\noindent
Both \cite{Bellin} and \cite{\wij} allow the set  $\{A_1, ..., A_n\}$ in \Wrule{} to be empty,
thus including implicitly \grmdiam. 
By uniformity with the formulation of the other rules, we require it to contain at least one formula.
Then, given the present formulation, \GCK{} and \GW{} are defined by extending \gtrei{} as follows:

\begin{center}
\begin{tabular}{l l l}
\GCK & := & \grmboxc{} + \grmdiam{} + \grulenbox{} + \Wrule \\

\GW & := & \grmboxc{} + \grmdiam{} + \grulenbox{} + \Wrule{} + \ginttrec{} + \grulendiam{}
\end{tabular}
\end{center}

Observe that \GW{} can be seen as an extension of our top calculus  \gMNboxC,
as it is \gMNboxC{} + \Wrule.
Instead, \GCK{} is not comparable with any our bimodal calculus,
as it contains rule \grulenbox{}  and doesn't contain \grulendiam,
what is never the case in the calculi of our cube.

\begin{theorem}[\cite{Bellin} for \GCK{}, \cite{\wij} for \GW]
Cut is admissible in \GCK{} and \GW.
Moreover, \GCK{} and \GW{} are equivalent with the corresponding axiomatisations.
\end{theorem}

Notice that having \Wrule{} istead of our ``weak interaction'' rules,
allows us to take \grulenbox{} and not \grulendiam{} (as in \GCK),
and still obtain a cut-free calculus.
If instead we take both \Wrule{} and \grulendiam{} (an in \GW), we need to 
take also \ginttrec{} in order to have the \cut{} rule admissible, as it is shown by the following derivation:

\begin{center}
\ax{$p, \neg p \seq \bot$}
\llab{\Wrule}
\uinf{$\Box p, \diam\neg p \seq \diam\bot$}
\ax{$\bot\seq$}
\rlab{\grulendiam}
\uinf{$\Box p, \diam\neg p, \diam\bot \seq$}
\rlab{\cut}
\binf{$\Box p, \diam\neg p \seq$}
\disp
\end{center}

\noindent
It is immediate to verify that the endsequent $\Box p, \diam\neg p \seq$
is derivable in $\textup{\GW} \setminus \{\textup{\ginttrec}\}$
if and only if the \cut{} rule is applied,
but it has a cut-free derivation in \GW{}
by applying \ginttrec{} to $p, \neg p \seq$.
Notice also that adding \ginttrec{} to the calculus preserve the equivalence with the axiomatisation, as 
\inttre{} is derivable from \axKdiam{}, \rmdiam{} and \axNdiam.

\subsection{Intuitionistic neighbourhood models for \CK{} and \HW}

We now define \intmodel s for \CK{} and \HW,
 and prove soundness and completeness
of the two systems.

\begin{definition}[\Intmodel s for \CK{} and \HW]
A \intmodel{} for \CK{} (\CK-model in the following) is any \intmodel{} in which  
$\nbox$ is supplemented, closed under intersection and contains the unit; $\ndiam$ is supplemented;
and such that:
\begin{center}
If $\alpha\in\nbox(w)$ and $\beta\in\ndiam(w)$, then $\alpha\cap\beta\in\ndiam(w)$ \qquad  (\wcondition).       
\end{center}
A \intmodel{} for \HW{} (\HW-model in the following) is any \intmodel{}  for \CK{} satisfying also the 
 condition of \Connuno{} ($\nbox(w)\subseteq\ndiam(w)$).     
\end{definition}

Notice that, as a consequence, function $\ndiam$ in \HW-models contains the unit.
We now show that logics \CK{} and \HW{} are sound and complete with respect to
the corresponding models.

\begin{theorem}[Soundness]
Logics \CK{} and \HW{} are sound with respect to 
\CK- and \HW-models, respectively.
\end{theorem}
\begin{proof}
We just consider axiom \axKdiam. 
Assume $w\Vd \Box(A\imp B)$ and $w\not\Vd\diam B$.
Then $[A\imp B]\in\nbox(w)$ and $\W\setminus[B]\in\ndiam(w)$.
By \wcondition,
$[A\imp B]\cap (\W\setminus[B])\in\ndiam(w)$.
Since $[A\imp B]\cap (\W\setminus[B])\subseteq (\W\setminus[A])$,
by supplementation
$\W\setminus[A]\in\ndiam(w)$;
therefore $w\not\Vd\diam A$.
\end{proof}

Completeness is proved as before by the canonical model construction.

\begin{lemma}\label{lemma canonical model CK}
Let the canonical models $\McCK$ for \CK, and $\McHW$ for \HW, be defined as in Lemma \ref{supplemented canonical model}. 
Then $\McCK$ and $\McHW$ are, respectively, a \CK-model and a \HW-model.
\end{lemma}
\begin{proof}
We show that both $\McCK$ and $\McHW$ satisfy the condition of \wcondition:
Assume $\alpha\in\nboxplus(X)$ and $\alpha\cap\beta\not\in\ndiamplus(X)$.
Then there are $A,B\in\lan$ such that $\up A\subseteq\alpha$, $\alpha\cap\beta\subseteq\Wc\setminus\up B$ and $\Box A,\diam B\in X$.
As a consequence, 
$\up A \cap\beta \subseteq \Wc\setminus\up B$,
that by standard properties of set inclusion
 implies $\beta\subseteq (\Wc\setminus\up A) \cup (\Wc\setminus\up B)
= \Wc\setminus\up (A\land B)$.
%\nb{Because $\alpha\cap\beta\subseteq\gamma$ implies $\beta\subseteq \gamma \cup \W\setminus\alpha$.}
Moreover,
since $(\Box A \land\diam B)\imp \diam(A\land B)$ is derivable 
(from $A \imp (B \imp A\land B)$, by \rmbox{} and \axKdiam{}),
we have $\diam(A\land B)\in X$.
Thus by definition, $\beta\notin\ndiamplus(X)$.
In addition, by Lemma \ref{supplemented canonical model} (iv)
$\McHW$ is also \connuno, as \inttre{} is derivable in \HW{}.
\end{proof}

\begin{theorem}[Completeness]
Logics \CK{} and \HW{} are complete with respect to 
\CK- and \HW-models, respectively.
\end{theorem}
\begin{proof}
Same proof of Theorem \ref{theorem completeness}, using Lemma \ref{lemma canonical model CK}.
\end{proof}

\subsection{Pre-existing semantics and direct proofs of equivalence} 

\subsubsection{Semantic equivalence for \HW}

We now consider pre-existing possible worlds semantics for systems \CK{} and \HW{},
and prove directly their equivalence with \intmodel s.
We begin with system \HW,
and consider the relational models by \wij{} \cite{\wij} as well as the neighbourhood models by Kojima \cite{Kojima}.

\begin{definition}[Relational models for \HW{} (\wij{} \cite{\wij})]\label{relational model for W} 
A relational model for \HW{} is a tuple 
$\M=\langle\W,\less,\R,\V\rangle$,
where $\W$, $\less$ and $\V$ are as in Definition \ref{\intmodel},
and $\R$ is any binary relation on $\W$.
The forcing relation $w\Vdr A$ 
is defined as $w\Vd A$ (Definition \ref{\intmodel}) for $A\equiv p, B\land C, B\lor C, B\imp C$;
and in the following way for modal formulas:

\vspace{0.2cm}
\begin{tabular}{l l l}
$w\Vdr\Box B$ & iff & for all $v\more w$, for all $u\in\W$, $v\R u$ implies $u\Vdr B$; \\
$w\Vdr\diam B$ & iff & for all $v\more w$, there is $u\in\W$ s.t. $v\R u$ and $u\Vdr B$. \\
\end{tabular}
\end{definition}

\begin{definition}[Kojima's neighbourhood models for \HW{} (Kojima \cite{Kojima})]
Kojima's neighbourhood models for \HW{} are tuples $\M=\langle \W, \less, \Nk, \V \rangle$, where
$\W$, $\less$ and $\V$ are, respectively, a non-empty set, a preorder on $\W$ and a
hereditary valuation function; 
and $\Nk$ is a neighbourhood function $\W \longrightarrow \pow(\pow(\W))$ such that:

\vspace{0.2cm}
\begin{tabular}{l}
$\bullet$ \ $w\less v$ implies $\Nk(v)\subseteq \Nk(w)$; \\
$\bullet$ \ $\Nk(w)\not=\emptyset$ for all $w\in\W$.
\end{tabular}

\vspace{0.2cm}
\noindent
The forcing relation $w\Vdk A$ is defined as usual for $A\equiv p,\bot, B\land C, B\lor C, B\imp C$; 
and for modal formulas it is defined as follows:

\vspace{0.2cm}
\begin{tabular}{l l l}
$w\Vdk \Box B$ & iff & for all $\alpha\in\Nk(w)$, for all $v\in\alpha$, $v\Vdk B$; \\
$w\Vdk \diam B$ & iff & for all $\alpha\in\Nk(w)$, there is $v\in\alpha$ s.t. $v\Vdk B$. \\
\end{tabular}
\end{definition}

\begin{theorem}[\wij{} \cite{\wij}, Kojima \cite{Kojima}]
Logic \HW{} is sound and complete w.r.t. relational models for \HW{},
as well as w.r.t. Kojima's models for \HW.
\end{theorem}

That relational, Kojima's and \intmodel s for \HW{} are equivalent is a corollary of the respective completeness theorems.
It is instructive, however, to prove the equivalence directly. 
A proof of equivalence of Kojima's and relational models is given in Kojima \cite{Kojima}.
Here we prove directly the equivalence of Kojima's and \intmodel s for \HW{}.
By combining the two proofs we then obtain direct transformations
between relational and \intmodel s. 

The following property will be considered in the proof of some of the next lemmas:
\begin{center}
For all $\alpha\in\ndiam(w)$, there is $\beta\in\ndiam(w)$ s.t. $\beta\subseteq\alpha$ and $\beta\subseteq\bigcap\nbox(w)$ \ (\wconditionbis).
\end{center}
This property is satisfied by all \intmodel s for \HW{} and for \CK,
as it follows from the intersection closure of $\nbox$ and the \wcondition.

\begin{lemma}\label{lemma W Kojima to neigh}
Let  $\Mk=\langle \W,\less,\Nk,\V\rangle$ be a \komodel{} for \HW, 
and let $\Mn$ be the model $\langle \W,\less,\nbox, \ndiam,\V\rangle$ where
$\W$, $\less$ and $\V$ are as in $\Mk$, and:

\vspace{0.2cm}
\begin{tabular}{l}
$\nbox(w)=\{\alpha\subseteq \W \mid \bigcup\Nk(w)\subseteq\alpha\}$; \\
$\ndiam(w)=\{\alpha\subseteq \W \mid \textup{there is }\beta\in\Nk(w) \textup{ s.t. } \beta\subseteq\alpha\}$. \\
\end{tabular}

\vspace{0.2cm}
\noindent
Then $\Mn$ is a \intmodel{} for \HW{} and is pointwise equivalent to $\Mk$.
\end{lemma}
\begin{proof}
It is immediate to verify that $\nbox$ and $\ndiam$ are supplemented and contain the unit; that $\nbox$ is closed under intersection;
and that $w\less v$ implies $\nbox(w)\subseteq\nbox(v)$ and $\ndiam(v)\subseteq\ndiam(w)$.
We show that $\Mn$ satisfies the other properties of \wmodel s.

(\Connuno) \ Assume $\alpha\in\nbox(w)$. 
Then $\bigcup\Nk(w)\subseteq\alpha$, and, since $\Nk(w)\not=\emptyset$, there is $\beta\in\Nk(w)$ such that $\beta\subseteq\alpha$.
Therefore $\alpha\in\ndiam(w)$.

(\wcondition) \ Assume $\alpha\in\nbox(w)$ and $\beta\in\ndiam(w)$.
Then $\bigcup\Nk(w)\subseteq\alpha$ and there is $\gamma\in\Nk(w)$ such that $\gamma\subseteq\beta$.
Thus $\gamma\subseteq\bigcup\Nk(w)$, which implies $\gamma\subseteq\alpha\cap\beta$.
Therefore $\alpha\cap\beta\in\ndiam(w)$.

By induction on $A$ we now prove that for all $A\in\lan$ and all $w\in\W$, 
\begin{center}
$\Mn, w \Vdn A$ \ iff \ $\Mk, w\Vdk A$.
\end{center}
We only consider the inductive cases $A\equiv\Box B, \diam B$.

$A\equiv\Box B$.
$\Mn, w \Vdn \Box B$ iff $[B]_{\Mn}\in\nbox(w)$ iff $\bigcup\Nk(w)\subseteq[B]_{\Mn}$ iff (i.h.) $\bigcup\Nk(w)\subseteq[B]_{\Mk}$
iff for all $\alpha\in\Nk(w)$, $\alpha\subseteq [B]_{\Mk}$ iff $\Mk, w \Vdk \Box B$.

$A\equiv\diam B$.
$\Mn, w \Vdn \diam B$ iff $\W\setminus [B]_{\Mn}\notin\ndiam(w)$ 
iff for all $\alpha\in\Nk(w)$, $\alpha\cap[B]_{\Mn}\not=\emptyset$
iff (i.h.) for all $\alpha\in\Nk(w)$, $\alpha\cap[B]_{\Mk}\not=\emptyset$ iff $\Mk, w \Vdk \diam B$.
\end{proof}

\begin{lemma}\label{lemma W neigh to Kojima}
Let  $\Mn=\langle \W,\less,\nbox, \ndiam,\V\rangle$ be a \intmodel{} for \HW,
and let  $\Mk$ be the model $\langle \W,\less,\Nk,\V\rangle$ where
$\W$, $\less$ and $\V$ are as in $\Mn$, and:

\vspace{0.2cm}
\begin{tabular}{l}
$\Nk(w)=\{\alpha\in\ndiam(w) \mid \alpha\subseteq \bigcap\nbox(w)\}$.
\end{tabular}

\vspace{0.2cm}
\noindent
Then $\Mk$ is a \komodel{} for \HW{} and is pointwise equivalent to $\Mn$.
\end{lemma}
\begin{proof}
First notice that $\Mk$ is a \komodel: 
By intersection closure, $\bigcap\nbox(w)\in\nbox(w)$, hence by \Connuno{}, $\bigcap\nbox(w)\in\ndiam(w)$.
Thus $\bigcap\nbox(w)\in\Nk(w)$, which implies $\Nk(w)\not=\emptyset$.
Moreover assume $w\less v$ and $\alpha\in\Nk(v)$. So $\alpha\in\ndiam(v)$ and $\alpha\subseteq\bigcap\nbox(v)$.
Since $\ndiam(v)\subseteq\ndiam(w)$ and $\nbox(w)\subseteq\nbox(v)$,
we have both $\alpha\in\ndiam(w)$ and $\alpha\subseteq\bigcap\nbox(w)$, therefore $\alpha\in\Nk(w)$. 

By induction on $A$ we show that for all $A\in\lan$ and all $w\in\W$, 
\begin{center}
$\Mn, w \Vdn A$ iff $\Mk, w\Vdk A$.
\end{center}
As before we only consider the inductive cases $A\equiv\Box B, \diam B$:

$A\equiv\Box B$.
$\Mk, w\Vdk\Box B$ iff for all $\alpha\in\Nk(w)$, $\alpha\subseteq[B]_{\Mk}$
iff (since $\bigcap\nbox(w)\in\Nk(w)$) $\bigcap\nbox(w)\subseteq[B]_{\Mk}$
iff (i.h.) $\bigcap\nbox(w)\subseteq[B]_{\Mn}$
iff (by properties of $\nbox(w)$) $[B]_{\Mn}\in\nbox(w)$
iff $\Mn, w\Vdn \Box B$.

$A\equiv \diam B$.
Assume $\Mk, w, \Vdk \diam B$. Then for all $\alpha\in\Nk(w)$, $\alpha\cap[B]_{\Mk}\not=\emptyset$,
and, by i.h., $\alpha\cap[B]_{\Mn}\not=\emptyset$.
Thus for all $\alpha\in\ndiam(w)$ s.t. $\alpha\subseteq\bigcap\nbox(w)$, $\alpha\cap[B]_{\Mn}\not=\emptyset$.
Let $\beta$ be any neighbourhood in $\ndiam(w)$. 
By \wconditionbis, there is $\gamma\subseteq\beta$ 
s.t. $\gamma\in\ndiam(w)$ and $\gamma\subseteq\bigcap\nbox(w)$.
Then $\gamma\cap[B]_{\Mn}\not=\emptyset$, which implies $\beta\cap[B]_{\Mn}\not=\emptyset$.
Therefore $\Mn, w \Vdn \diam B$.
Now assume $\Mn, w \Vdn \diam B$. Then for all $\alpha\in\ndiam(w)$, $\alpha\cap[B]_{\Mn}\not=\emptyset$.
Thus for all $\alpha\in\Nk(w)$, $\alpha\cap[B]_{\Mn}\not=\emptyset$, and, by i.h.,
$\alpha\cap[B]_{\Mk}\not=\emptyset$.
Therefore $\Mk, w\Vdk \diam B$.
\end{proof}

\begin{theorem}\label{W Kojima to neigh}
A formula $A$ is valid in \komodel s for \HW{} if and only if it is valid in \intmodel s for \HW{}.
\end{theorem}
\begin{proof}
By Lemmas \ref{lemma W Kojima to neigh} and \ref{lemma W neigh to Kojima}.
If a \komodel{} for \HW{} falsifies $A$, 
then there is a \intmodel{}  for \HW{} that falsifies $A$; and \emph{vice versa}
if a \intmodel{} for \HW{} falsifies $A$, 
then there is a \komodel{} for \HW{} that falsifies $A$.
\end{proof}

Given the previous lemmas and Theorems 4.3 and 4.7 in Kojima \cite{Kojima}, %(pp. 97, 98),
we can also see  how to obtain an equivalent relational model starting from a \intmodel{} for \HW, and \emph{vice versa}.

\begin{lemma}\label{lemma W rel to neigh}
Let $\Mr=\langle \W,\less,\R,\V\rangle$ be a relational model for \HW,
and let $\rel(w)=\{v \mid w\R v\}$.
We define the neighbourhood model
$\Mn=\langle \W,\less,\nbox,\ndiam,\V\rangle$  by taking $\W$, $\less$, $\V$ as in $\Mr$, 
and the following neighbourhood functions:

\vspace{0.2cm}
\begin{tabular}{l}
$\nbox(w)=\{\alpha\subseteq \W \mid \textup{for all } v\more w, \rel(v)\subseteq\alpha\}$; \\
$\ndiam(w)=\{\alpha\subseteq \W \mid \textup{there is } v\more w \textup{ s.t. } \rel(v)\subseteq\alpha\}$.
\end{tabular}

\vspace{0.2cm}
Then $\Mn$ is a \intmodel{} for \HW, and it is pointwise equivalent to $\Mr$.
\end{lemma}

\begin{lemma}\label{lemma W neigh to rel}
Let $\Mn=\langle \W, \less, \nbox, \ndiam, \V \rangle$ be a \intmodel{} for \HW.
The relational model $\Mstar= \langle \Wstar, \lessstar, \Rstar, \Vstar \rangle$ is defined as follows:

\vspace{0.4cm}
\begin{tabular}{l }
\vspace{0.2cm}
$\bullet$ \ $\Wstar$  =  $\{(w, \alpha)$ $\mid$ $w\in\W$, $\alpha\in\ndiam(w)$,  and  $\alpha\subseteq\bigcap\nbox(w)\}$; \\
\vspace{0.2cm}
$\bullet$ \ $(w, \alpha)\lessstar (v, \beta)$ \ iff \ $w\less v$; \\
\vspace{0.2cm}
$\bullet$ \ $(w, \alpha)\Rstar (v, \beta)$ \ iff \ $v\in\alpha$; \\
\vspace{0.2cm}
$\bullet$ \ $\Vstar((w, \alpha))=\{p \mid p\in\V(w)\}$  for all $w\in\W$.
\end{tabular}

\vspace{0.4cm}
\noindent
Then $\Mstar$ is a relational model for \HW.
Moreover, for all $A\in\lan$ and $w\in\W$, the following claims are equivalent:

\begin{itemize}
\item[1)] $\Mn, w \Vdn A$.

\item[2)] For all $(w, \alpha)\in\Wstar$, $\Mstar, (w, \alpha)\Vdr A$.

\item[3)] There is $(w, \alpha)\in\Wstar$ such that $\Mstar, (w, \alpha)\Vdr A$.
\end{itemize}
\end{lemma}

\begin{theorem}\label{W neigh to rel}
A formula $A$ is valid in relational models for \HW{} if and only if
it is valid in \intmodel s for \HW.
\end{theorem}
\begin{proof}
By Lemma \ref{lemma W rel to neigh} and Lemma \ref{lemma W neigh to rel}.
A direct proof of the two lemmas is left to the reader.
\end{proof}

\subsubsection{Semantic equivalence for \CK}

We now present the relational models for \CK{} by Mendler and de Paiva \cite{Mendler1},
and prove directly their equivalence with \intmodel s.
Relational models for \CK{} are defined by enriching \wij's models for \HW{}
with inconsistent (or ``fallible'') worlds
(\emph{i.e.} worlds satisfying $\bot$) as follows.

\begin{definition}[Relational models for \CK{}]\label{relational model for CK}  %(Mendler and de Paiva \cite{Mendler1})
Relational models for \CK{} are defined exactly as relational models for \HW{} (Definition \ref{relational model for W}),
except that the standard forcing relation for $\bot$ ($w\not\Vdr\bot$ )
is replaced by the following ones:

\vspace{0.2cm}
\begin{tabular}{l}
If $w\Vdr\bot$, \ then for all $v$, $w\less v$ or $w\R v$ implies $v\Vdr\bot$; \\
If $w\Vdr\bot$, \ then $w\Vdr p$ for all propositional variables $p\in\lan$.
\end{tabular}
\end{definition}

Observe that fallible worlds are related through $\less$ and $\R$  only to other fallible worlds.
Moreover, the above definition  preserves the validity of $\top$ and $\bot \imp A$, for all $A$.

\begin{theorem}[Mendler and de Paiva \cite{Mendler1}]
Logic \CK{} is sound and complete w.r.t. relational models for \CK.
\end{theorem}

In order to prove the equivalence between relational and \intmodel s for \CK,
we consider transformations of models that are relatively similar to those in 
Lemmas \ref{lemma W rel to neigh} and \ref{lemma W neigh to rel}.
However, the transformations are now a bit more complicated
due to the presence of inconsistent worlds.

\begin{lemma}\label{lemma CK rel to neigh}
Let $\Mr=\langle \W,\less,\R,\V\rangle$ be a relational model for \CK.
Moreover,
for all $w\in\W$, let $\rel(w)=\{v \mid w\R v\}$.
We denote with $\W\pos$ the set $\{w\in\W \mid \Mr, w\not\Vdr\bot\}$ ({\it i.e.}~the set of consistent worlds of $\Mr$),
and for all $\alpha\subseteq\W$, we denote with $\alpha\pos$ the set $\alpha\cap\W\pos$.

We define the neighbourhood model $\Mn=\langle \W\pos,\less\pos,\nbox,\ndiam,\V\pos\rangle$, 
where $\less\pos$ and $\V\pos$ are the restrictions to $\W\pos$ of $\less$ and $\V$, 
and $\nbox$, $\ndiam$ are the following neighbourhood functions:

\vspace{0.2cm}
\begin{tabular}{l}
$\nbox(w)=\{\alpha\pos\subseteq \W \mid \textup{for all } v\more w, \rel(v)\subseteq\alpha\}$; \\
$\ndiam(w)=\{\alpha\pos\subseteq \W \mid \textup{there is } v\more w \textup{ s.t. } \rel(v)\subseteq\alpha\pos\}$.
\end{tabular}

\vspace{0.2cm}
\noindent
Then $\Mn$ is a \intmodel{} for \CK{}.
Moreover, for all $A\in\lan$ and $w\in\W\pos$, 
\begin{center}
$\Mn,w\Vdn A$ \ iff \ $\Mr,w\Vdr A$. 
\end{center}
\end{lemma}
\begin{proof}
It is imediate to verify that $\Mn$ is a \intmodel{} for \CK.
In particular, for the \wcondition,
assume $\alpha\pos\in\nbox(w)$ and $\beta\pos\in\ndiam(w)$.
Then there is $v\more w$ s.t. $\rel (v)\subseteq\beta\pos$; thus $\rel (v)\subseteq\alpha$.
Then $\rel (v)\subseteq \alpha\cap\beta\pos = (\alpha\cap\beta)\pos$.
Therefore $(\alpha\cap\beta)\pos = \alpha\pos\cap\beta\pos \in\ndiam(w)$.

We now prove that for all $w\in\W\pos$, $\Mn,w\Vdn A$ if and only if $\Mr,w\Vdr A$.
This is equivalent to say that $[A]_{\Mn} = [A]_{\Mr}\pos$.
As usual we only consider the modal cases.

$A\equiv \Box B$. Let $w\in\W\pos$.
$\Mn, w \Vdn \Box B$ iff $[B]_{\Mn}\in\nbox(w)$
iff (i.h.) $[B]_{\Mr}\pos\in\nbox(w)$ 
iff for all $v\more w$, $\rel(v)\subseteq [B]_{\Mr}$ 
iff $\Mr, w \Vdr \Box B$.

$A\equiv \diam B$.
Assume $\Mr, w \Vdr \diam B$ and $w\in\W\pos$.
Then for all $v\more w$, there is $u\in\W$ s.t. $v\R u$ and $\Mr, u \Vdr B$.
Thus for all $v\more w$, $\R(v)\not\subseteq \W\setminus [B]_{\Mr}$, which in particular implies $\R(v)\not\subseteq (\W\setminus [B]_{\Mr})\pos$.
Moreover, $ (\W\setminus [B]_{\Mr})\pos = \W\pos\setminus [B]_{\Mr}\pos = \textup{(i.h.) } \W\pos\setminus [B]_{\Mn}$.
Then $\W\pos\setminus [B]_{\Mn}\notin\ndiam(w)$, 
therefore $\Mn, w \Vdn \diam B$.
Now assume $\Mn, w \Vdn \diam B$. 
Then $\W\pos\setminus [B]_{\Mn}\notin\ndiam(w)$.
This implies that for all $v\more w$, $\R(v)\not\subseteq\W\pos\setminus [B]_{\Mn}$;
that is, there is $u\in\W$ s.t. $v\R u$ and $u\notin\W\pos\setminus [B]_{\Mn}$.
Thus $u\notin\W\pos$ or $u\in[B]_{\Mn}$.
If $u\notin\W\pos$, then $\Mr, u\Vdr \bot$, hence $\Mr, u\Vdr B$.
If $u\in [B]_{\Mn}$, by i.h. $u\in[B]_{\Mr}\pos$, thus $\Mr, u\Vdr B$.
Therefore $\Mr, w\Vdr \diam B$.
\end{proof}

\begin{lemma}\label{lemma CK neigh to rel}
Let $\Mn=\langle \W, \less, \nbox, \ndiam, \V \rangle$ be a \intmodel{} for \CK,
and take $\f\notin\W$.
The relational model $\Mstar= \langle \Wstar, \lessstar, \Rstar, \Vstar \rangle$ is defined as follows:

\vspace{0.4cm}
\begin{tabular}{l l}
\vspace{0.1cm}
$\bullet$ \ $\Wstar$  = & $\{(w, \alpha)$ $\mid$ $w\in\W$, $\ndiam(w)\not=\emptyset$, $\alpha\in\ndiam(w)$,  and  $\alpha\subseteq\bigcap\nbox(w)\}$ \\
\vspace{0.1cm}
& $\cup$ \ $\{(v, \bigcap\nbox(v)\cup\{\f\}) \mid v\in\W \textup{ and } \ndiam(v)=\emptyset\}$ \\
&  $\cup$ \ $\{(\f, \{\f\})\}$;
\end{tabular}

\vspace{0.2cm}
\begin{tabular}{l}
\vspace{0.2cm}
$\bullet$ \ $(w, \alpha)\lessstar (v, \beta)$ \ iff \ $w\less v$ \  or \ $w,v=\f$; \\
\vspace{0.2cm}
$\bullet$ \ $(w, \alpha)\Rstar (v, \beta)$ \ iff \ $v\in\alpha$; \\
\vspace{0.2cm}
$\bullet$ \ $\Vstar((w, \alpha))=\{p \mid p\in\V(w)\}$  for all $w\in\W$; \ and $\Vstar((\f, \{\f\}))= \atm$; \\
$\bullet$ \ $\Mstar, (\f, \{\f\}) \Vdr \bot$.
\end{tabular}

\vspace{0.4cm}
\noindent
Then $\Mstar$ is a relational model for \CK.
Moreover, for all $A\in\lan$ and $w\in\W$, the following claims are equivalent:

\begin{itemize}
\item[1)] $\Mn, w \Vdn A$.

\item[2)] For all $(w, \alpha)\in\Wstar$, $\Mstar, (w, \alpha)\Vdr A$.

\item[3)] There is $(w, \alpha)\in\Wstar$ such that $\Mstar, (w, \alpha)\Vdr A$.
\end{itemize}
\end{lemma}

\begin{proof}
It is immediate to show $\Mstar$ is a relational model for \CK{}, 
in particular the conditions on inconsistent worlds are satisfied.
We prove by induction on $A$ that points 1), 2) and 3) are equivalent.
As usual we only consider the inductive cases $A\equiv \Box B, \diam B$.

\vspace{0.2cm}
\noindent
$\bullet$ \ $A\equiv \Box B$.
\begin{itemize}
\item[$-$]  
1) implies 2). 
Assume $\Mn, w \Vdn \Box B$.
Then $[B]_{\Mn}\in\nbox(w)$, that implies $\bigcap\nbox(w)\subseteq[B]_{\Mn}$.
Let $(w, \alpha)\in\Wstar$, and $(w,\alpha)\lessstar (v, \beta)$.
Then $w\less v$, so $\bigcap\nbox(v)\subseteq\bigcap\nbox(w)$.
We distinguish two cases:

\begin{itemize}
\item[($a$)] $\f\in\beta$. 
Then $(v, \beta)\Rstar (u, \gamma)$ implies $u\in\bigcap\nbox(v)$ or $u=\f$.

If $u=\f$, then $(u, \gamma) = (\f, \{\f\})$, so $\Mstar, (u, \gamma) \Vdr B$.

If $u\in\bigcap\nbox(v)$, then $u\in[B]_{\Mn}$.
By i.h. we have $\Mstar, (u, \gamma) \Vdr B$
for all $\gamma$ s.t. $(u, \gamma)\in\Wstar$.

\item[($b$)] $\f\notin\beta$. 
Then $\beta\subseteq\bigcap\nbox(v)$, thus $\beta\subseteq[B]_{\Mn}$.
Let $(v, \beta)\Rstar (u, \gamma)$. Then $u\in\beta$, so $\Mn, u\Vdn B$.
By i.h. we have $\Mstar, (u, \gamma) \Vdr B$. 
\end{itemize}

By ($a$) and ($b$) we have that for all $(v, \beta)\morestar (w, \alpha)$ 
and all $(u, \gamma)$ s.t. $(v, \beta)\Rstar (u, \gamma)$, $\Mstar, (u, \gamma) \Vdr B$. 
Therefore for all $\alpha$ s.t. $(w, \alpha)\in\Wstar$, 
$\Mstar, (w, \alpha) \Vdr \Box B$.

\item[$-$]  
2) implies 3). 
Immediate because for all $w\in\W$ there is $\alpha$ s.t. $(w, \alpha)\in\Wstar$.

\item[$-$] 
3) implies 1). 
Assume $\Mstar, (w, \alpha) \Vdr \Box B$ for an $\alpha$ s.t. $(w, \alpha)\in\Wstar$.
Then for all $(v, \beta)\morestar (w, \alpha)$ and all $(u, \gamma)$ s.t. $(v, \beta)\Rstar (u, \gamma)$, 
$\Mstar, (u, \gamma) \Vdr B$.
Thus in particular, for all $\delta$ s.t. $(w, \delta)\in\Wstar$, 
for all $(u, \gamma)$ s.t. $(w, \delta)\Rstar (u, \gamma)$,  $\Mstar, (u, \gamma) \Vdr B$.
Take any world $z\in\bigcap\nbox(w)$.
There exists $\gamma$ s.t. $(z,\gamma)\in \Wstar$.
Then $(w, \bigcap\nbox(w))\Rstar (z,\gamma)$ or $(w, \bigcap\nbox(w)\cup \{f\})\Rstar (z,\gamma)$
(depending on whether $\ndiam(w)\not=\emptyset$ or  $\ndiam(w)=\emptyset$;
in the first case $\bigcap\nbox(w)\in\ndiam(w)$).
Thus $\Mstar, (z, \gamma) \Vdr B$;
and by i.h., $\Mn, z \Vdn B$.
So $\bigcap\nbox(w)\subseteq [B]_{\Mn}$, which implies $[B]_{\Mn}\in\nbox(w)$.
Therefore $\Mn, w \Vdn \Box B$.
\end{itemize}

\noindent
$\bullet$ \ $A\equiv \diam B$.
\begin{itemize}
\item[$-$] 
1) implies 2). 
Assume $\Mn, w \Vdn \diam B$, and let $(w, \alpha)\in\Wstar$ and $(w, \alpha) \lessstar (v, \beta)$.
We distinguish two cases:

\begin{itemize}
\item[($a$)] $\f\in\beta$. Then $(y, \beta)\Rstar\ff $, and $\Mstar, \ff\Vdr B$.

\item[($b$)] $\f\notin\beta$. Then $\beta\in\ndiam(y)$, so $\beta\in\ndiam(y)$.
By $\Mn, w \Vdn \diam B$, we have that for all $\gamma\in\ndiam(w)$, $\gamma\cap[B]_{\Mn}\not=\emptyset$;
thus $\beta\cap[B]_{\Mn}\not=\emptyset$.
Then there is $u\in\beta$ s.t. $\Mn, u\Vdn B$.
By i.h., for all $\delta$ s.t. $(u, \delta)\in\Wstar$, $\Mstar, (u, \delta)\Vdr B$.
Moreover, there is $\epsilon$ s.t. $(u, \epsilon)\in\Wstar$.
Thus $(v, \beta)\Rstar (u, \epsilon)$ and $\Mstar, (u, \epsilon)\Vdr B$.
\end{itemize}

By ($a$) and ($b$) we have that for all $(v, \beta)\morestar (w, \alpha)$,
there is $(u, \gamma)$ s.t. $(v, \beta)\Rstar (u, \gamma)$ and $\Mstar, (u, \gamma)\Vdr B$.
Therefore, for all $\alpha$ s.t. $(w, \alpha)\in\Wstar$, $\Mstar, (w, \alpha) \Vdr \diam B$.

\item[$-$] 
2) implies 3). 
Immediate because for all $w\in\W$ there is $\alpha$ s.t. $(w, \alpha)\in\Wstar$.

\item[$-$] 
3) implies 1). 
Assume $\Mstar, (w, \alpha)\Vdr \diam B$ for a $\alpha$ s.t. $(w, \alpha)\in\Wstar$.
Then for all $(v, \beta)\morestar (w, \alpha)$, 
there is $(u, \gamma)$ s.t. $(v, \beta)\Rstar (u, \gamma)$ and $\Mstar, (u, \gamma)\Vdr B$.
Thus in particular, for all $\delta$ s.t. $(w, \delta)\in\Wstar$, 
there is $(u, \gamma)$ s.t. $(w, \delta)\Rstar (u, \gamma)$ and $\Mstar, (u, \gamma)\Vdr B$.
We distinguish two cases:

\begin{itemize}
\item[($a$)] $\f\in\delta$ for a $(w, \delta)\in\Wstar$. Then $\ndiam(w)=\emptyset$, so $\Mn, w \Vdn \diam B$.

\item[($b$)] $\f\notin\delta$ for all $(w, \delta)\in\Wstar$. 
Then by i.h. we have that for all $(w, \delta)\in\Wstar$, 
there is $(u, \gamma)$ s.t. $(w, \delta)\Rstar (u, \gamma)$ and $\Mn, u\Vdn B$.
So $u\in\delta$.
This means that for all $\delta\in\ndiam(w)$ s.t. $\delta\subseteq\bigcap\nbox(w)$, $\delta\cap[B]_{\Mn}\not=\emptyset$.
Then by \wconditionbis, 
we have that for all $\epsilon\in\ndiam(w)$, $\epsilon\cap[B]_{\Mn}\not=\emptyset$.
Therefore $\Mn, w \Vdn \diam B$.
\end{itemize}
\end{itemize}
\end{proof}

\begin{theorem}
A formula $A$ is valid in relational models for \CK{}
if and only if it is valid in \intmodel s for \CK.
\end{theorem}
\begin{proof}
Assume $A$ not valid in relational models for \CK{}.
Then there are a relational model $\Mr$ and a world $w$ such that
$\Mr, w \not\Vdr A$.
World $w$ is consistent (\emph{i.e.} $\Mr, w \not\Vdr \bot$) as inconsistent worlds satisfy all formulas.
Then by Lemma \ref{lemma CK rel to neigh}, there is a \intmodel{} $\Mn$ for \CK{} such that $\Mn, w \not\Vdn A$.

Now assume $A$ not valid in \intmodel s for \CK.
Then there are $\Mn$ and $w$ such that $\Mn, w \not\Vdn A$.
By Lemma \ref{lemma CK neigh to rel}, there are 
a relational model $\Mstar$ and a world $(w, \alpha)$ such that $\Mstar, (w, \alpha)\not\Vdr A$.
\end{proof}

\section{Conclusion and further work}
This work represents the initial step towards a general investigation of non-normal modalities with an intuitionistic base.
We have defined a new family of \intlogic s 
that can be seen as intuitionistic counterparts of classical non-normal modal logics.
In particular, we have defined 12 monomodal logics 
-- 8 logics with $\Box$ modality and 4 logics with $\diam$ modality -- and 24 bimodal logics. 
For each of them we have provided both a Hilbert axiomatisation and a cut-free sequent calculus.
All logics are decidable
and contain some of the modal axioms characterising the classical cube.
In addition, bimodal logics contain interactions between the modalities that 
can be seen as ``weak duality principles'', and express under which conditions two formulas $\Box A$ and $\diam B$ are jointly inconsistent.
On the basis of the different strength of such interactions we identify different intuitionistic counterparts of a given classical logic.

Subsequently, we have given a modular semantic characterisation of the logics by means of so-called \cupledintmodel s.
The models contain an order relation and two neighbourhood functions
handling the modalities separately. 
For the two functions we consider the standard properties of neighbourhood models,
moreover they can be combined in different ways reflecting the possible interactions between $\Box$ and $\diam$.   
Through a filtration argument we have also proved that most of the logics enjoy the finite model property.
Our semantics turned out to be a versatile tool to analyse intuitionistic non-normal modal logics,
which is capable of capturing further well-known \intbilogic s as Constructive \K{} and the propositional fragment of \wij's \CCDL.

Our results can be extended in several directions. First of all  we can study further extensions of the cube by axioms  analogous to the standard modal ones such as \axT, \axD, \axquattro, \axcinque, \emph{etc.} 
(some cases have already been considered by Witczak \cite{Witczak2}).
Furthermore, we can study  computational and proof-theoretical properties such as complexity bounds and interpolation. 
To this regard we plan to develop  sequent calculi   with invertible rules  and that allow for  
direct countermodel extraction.

From the semantical side we intend to investigate  whether it can be given a semantic characterisation of axiom \axCdiam{}, 
that to our knowledge has not been captured yet.

Finally, it would be interesting to see whether these logics, similarly to \CK{}, can be given a type-theoretical interpretation by a suitable extension of the typed lambda-calculus. 
All of this will be part of our future research.


\begin{thebibliography}{99}
\bibitem{Bellin} Bellin, G., V.~de Paiva, and E.~Ritter, Extended Curry-Howard Correspondence for a Basic Constructive Modal Logic,
in: Proceedings of Methods for Modalities, 2001.

\bibitem{Chellas} Chellas, B. F., Modal Logic: An Introduction, Cambridge University Press, 1980.

\bibitem{aiml} Dalmonte, T., N.~Olivetti, and S.~Negri, Non-normal modal logics: bi-neighbourhood semantics and its labelled calculi,
in: Proceedings of AiML 12, 2018. %, pp.~159--178.

\bibitem{Fairtlough} Fairtlough, M., and M.~Mendler, Propositional Lax Logic, Information and Computation 137(1) (1997), pp.~1--33.

\bibitem{FischerServi} Fischer Servi, G., On modal logic with an intuitionistic base, Studia Logica, 36(3) (1977), pp.~141--149.

\bibitem{FischerServi} Fischer Servi, G., Semantics for a class of intuitionistic modal calculi,
in: Italian studies in the philosophy of science, Springer, 1980, pp.~59--72.

\bibitem{Fitch} Fitch, F.B., Intuitionistic modal logic with quantifiers, Portugaliae Mathematica, 7(2) (1948), pp.~113--118.

\bibitem{Gabbay} Gabbay, D.M., A. Kurucz, F. Wolter, and M. Zakharyaschev, Many-dimensional Modal Logics: Theory and Applications, Studies in Logic and The Foundations of Mathematics, vol. 148. North Holland Publishing Company, 2003.

\bibitem{Galmiche} Galmiche, D., and Y.~Salhi, Tree-Sequent Calculi and Decision Procedures for Intuitionistic Modal Logics, 
Journal of Logic and Computation, 28(5) (2018), pp.~967--989.

\bibitem{Galmiche2} Galmiche, D., and Y.~Salhi, Label-free Natural Deduction Systems for Intuitionistic and Classical modal logics, Journal of Applied Non-Classical Logics, 20(4) (2010), pp.~373--421.

\bibitem{Goldblatt} Goldblatt, R.I., Grothendieck topology as geometric modality, Mathematical Logic Quarterly, 27(31‐35) (1981), pp.~495--529.

\bibitem{Hilken} Hilken, B.P.,  Topological duality for intuitionistic modal algebras, Journal of Pure and Applied Algebra 148 (2) (2000), pp.~171--189.

\bibitem{Kojima} Kojima, K., Relational and Neighborhood Semantics for Intuitionistic Modal Logic, Reports on Mathematical Logic,
47 (2012), pp.~87–113.

\bibitem{Kojima2} Kojima, K., and A. Igarashi, Constructive linear-time temporal logic: Proof systems and Kripke semantics,
Information and Computation, 209(12) (2011), pp.~1491-1503.

\bibitem {Lavendhomme} Lavendhomme, R., and T. Lucas, Sequent Calculi and Decision Procedures for Weak Modal Systems,  Studia Logica, 66 (2000), pp.~121--145.

\bibitem{Lellmann} Lellmann, B., and E.~Pimentel, Proof search in nested sequent calculi, in:
Logic for Programming, Artificial Intelligence, and Reasoning, Springer, 2015, pp. 558--574.

\bibitem{Marin} Marin, S., and L. Stra\ss burger, Label-free Modular Systems for Classical and Intuitionistic Modal Logics,
in: Proceedings of AiML 10, 2014.

\bibitem {McNamara} McNamara, P., Deontic Logic. In: Gabbay and Woods (eds.), Handbook of the History of Logic, vol. 7 (2006), Elsevier, pp. 197--288.

\bibitem{Mendler1} Mendler, M., and V.~de Paiva, Constructive CK for Contexts, in:
Proceedings of CONTEXT05, Stanford, 2005.

\bibitem{Mendler2} Mendler, M., and S.~Scheele, Cut-free Gentzen calculus for multimodal CK. Information and Computation, 209(12) (2011), 
pp.~1465--1490.

%\bibitem{dePaiva} de Paiva, V., and B.~Pientka, Intuitionistic Modal Logic and Applications (IMLA 2008), 2011, pp.~1435--1436.

\bibitem{Pacuit} Pacuit, E., Neighborhood Semantics for Modal Logic, Springer, 2017.

\bibitem{Scott} Scott, D., Advice in modal logic, in: K. Lambert (\emph{ed.}), Philosophical Problems in Logic,
D. Reidel Publishing Company, 1970, pp.~143--173. 


\bibitem{Simpson} Simpson, A. K., The Proof Theory and Semantics of Intuitionistic Modal Logic. PhD thesis, 
School of Informatics, University of Edinburgh, 1994.

\bibitem{Stassburger} Stra\ss burger, L., Cut Elimination in Nested Sequents for Intuitionistic Modal Logics, 
in: International Conference on Foundations of Software Science and Computational Structures, Springer, 2013, pp.~209--224.

\bibitem{Stewart} Stewart, C., V. de Paiva, and N. Alechina, Intuitionistic Modal Logic: A 15-year retrospective,
Journal of Logic and Computation, 28(5) (2018), pp.~873--882.


\bibitem{Troelstra} Troelstra, A.S., and H. Schwichtenberg. Basic proof theory, Cambridge University Press, 2000.

\bibitem{Wijesekera} Wijesekera, D., Constructive modal logics I, Annals of Pure and Applied Logic, 50 (1990), pp.~271--301.

\bibitem{Witczak} Witczak, T., Intuitionistic modal logic based on neighborhood semantics without superset axiom,
arXiv preprint arXiv:1707.03859 (2017).

\bibitem{Witczak2} Witczak, T., Simple example of weak modal logic based on intuitionistic core, arXiv preprint arXiv:1806.09443 (2018).

\bibitem{Wolter} Wolter, F., and M. Zakharyaschev, Intuitionistic modal logic. 
In: Logic and Foundations of Mathematics, Springer, 1999, pp.~227--238. 

\bibitem{Wright} von Wright, G. H., Norm and Action: A Logical Enquiry, Routledge, 1963.

\end{thebibliography}
\end{document}